\documentclass[11pt]{article}
\usepackage[letterpaper, left=1in, right=1in, top=1in,bottom=1in]{geometry}
\usepackage{microtype}
\usepackage{parskip}

\usepackage{booktabs} 

\usepackage[utf8]{inputenc}

\usepackage{geometry}
\usepackage{amsmath,amsfonts,graphpap,amscd,mathtools,mathrsfs,graphicx,lscape,enumitem,dsfont,bm,url,subfigure}
\usepackage{epsfig,amssymb,amstext,xspace}
\usepackage{algorithm,algorithmic,comment}

\usepackage{amsthm}
\usepackage{thmtools,thm-restate}
\newtheorem{theorem}{Theorem}
\newtheorem*{theorem*}{Theorem}
\newtheorem{lemma}[theorem]{Lemma}
\newtheorem{prop}[theorem]{Proposition}

\newtheorem{corollary}[theorem]{Corollary}
\newtheorem{defn}[theorem]{Definition}

\def\squareforqed{\hbox{\rule{2.5mm}{2.5mm}}}

\def\p{\mathbf{p}}
\def\q{\mathbf{q}}
\def\f{\mathbf{f}}

\def\x{\mathbf{x}}
\def\X{\mathbf{X}}

\def\y{\mathbf{y}}
\def\1{\mathbf{1}}
\def\z{\mathbf{z}}
\def\w{\mathbf{w}}

\def\r{\mathbf{r}}
\def\Snm{\mathcal{S}_{n,m}}

\def\QED{\ifmmode\squareforqed 
  \else{\nobreak\hfil   
    \penalty50                 
    \hskip1em                  
    \null                      
    \nobreak                   
    \hfil                      
    \squareforqed              
    \parfillskip=0pt           
    \finalhyphendemerits=0     
    \endgraf}                  
  \fi}

\newcommand{\R}{\ensuremath{\mathbb R}}

\newcommand{\EE}{\ensuremath{\mathbb{E}}}

\newcommand{\PP}{\ensuremath{\mathbb{P}}}

\usepackage{color}              
\usepackage{auxiliary}

\usepackage[suppress]{color-edits}

\addauthor{df}{brown}
\addauthor{tl}{cyan}
\addauthor{sb}{magenta}

\begin{document}

\title{Pricing and Optimization in Shared Vehicle Systems:\\ An Approximation Framework}

\author{
Siddhartha Banerjee \thanks{Cornell University, \texttt{sbanerjee@cornell.edu}. 
}
\and Daniel Freund \thanks{MIT, \texttt{ dfreund@mit.edu}.
}
\and Thodoris Lykouris \thanks{Microsoft Research NYC, \texttt{thlykour@microsoft.com}. 
}}

\date{First version: August 2016\\Current version: May 2021%
\footnote{The current version represents the content that will appear in Operations Research. A one-page abstract of the paper appeared at the 18th ACM Conference on Economics and Computation (EC 2017).}}


\maketitle
\begin{abstract}
Optimizing shared vehicle systems (bike/scooter/car/ride-sharing) is more challenging compared to traditional resource allocation settings due to the presence of \emph{complex network externalities} -- changes in the demand/supply at any location affect future supply throughout the system within short timescales. These externalities are well captured by steady-state Markovian models, which are therefore widely used to analyze such systems. However, using such models to design pricing and other control policies is computationally difficult since the resulting optimization problems are high-dimensional and non-convex.

To this end, we develop a \emph{rigorous approximation framework} for shared vehicle systems, providing a unified approach for a wide range of controls (pricing, matching, rebalancing), objective functions (throughput, revenue, welfare), and system constraints (travel-times, welfare benchmarks, posted-price constraints).
Our approach is based on the analysis of natural convex relaxations, and obtains as special cases existing approximate-optimal policies for limited settings, asymptotic-optimality results, and heuristic policies.
The resulting guarantees are non-asymptotic and parametric, and provide operational insights into the design of real-world systems. In particular, for any shared vehicle system with $n$ stations and $m$ vehicles, our framework obtains an approximation ratio of $1+(n-1)/m$, which is particularly meaningful when $m/n$, the average number of vehicles per station, is large, as is often the case in practice.
\end{abstract}
\thispagestyle{empty}
\newpage

\setcounter{page}{1}

\setcounter{footnote}{0}

\section{Introduction}
\label{sec:introduction}

Shared vehicle systems, such as those for bike-sharing (Citi Bike in NYC, V\'{e}lib in Paris), scooter-sharing (Lime, Bird) car-sharing (car2go, Zipcar), and ride-sharing (Uber, Lyft), are fast becoming essential components of the urban transit infrastructure. In such systems, customers have access to a collection of personal transportation vehicles, which can be engaged anytime (subject to vehicle availability) and between a large number of source-destination locations. In bike/scooter/car-sharing, the vehicles are operated by the customers themselves, while in ride-sharing they are operated by independent drivers.

All vehicle sharing systems experience inefficiencies due to \emph{limited supply} (vehicles) and \emph{demand heterogeneity} across time and space. 
These inefficiencies, however, can often be greatly reduced by \emph{rebalancing} the demand and/or supply. 
Pricing has traditionally been the main tool to balance demand and supply in settings with limited supply and heterogeneous demand: for instance, in limited-item auctions or airplane/hotel reservations~\cite{gallego1994optimal,hartline2013}. Ride-sharing platforms have also long utilized dynamic pricing, and today, many bike-sharing platforms are experimenting with location-based incentive schemes~\cite{chung2018bike}. Shared vehicle platforms also enable other means of rebalancing, such as repositioning empty vehicles, or redirecting customers to nearby stations, though the precise set of deployed rebalancing tools depends on the kind of system: product design, regulatory constraints, or operational capabilities may limit or constrain the set of available rebalancing tools.
Moreover, these tools can be used towards achieving different objectives -- examples include revenue/welfare/throughput maximization, multi-objective settings, etc.

In contrast to other limited supply settings, however, shared vehicle settings are more challenging due to two unique features. The first is the presence of \emph{spatial and temporal supply externalities}: whenever a customer engages a vehicle, this not only decreases the instantaneous availability at the source location, but also affects the future availability at all other locations in the system. 
The other distinguishing feature is the high frequency of events (passenger arrivals/rides) in such settings. This necessitates treating the problem as an infinite-horizon control problem, as opposed to finite-horizon dynamic programming approaches used in traditional transportation and revenue management settings~\cite{gallego1994optimal,adelman2007price}. 
The high frequency of events tends to drive such systems into operating under a dynamic equilibrium state, jointly determined by the demand and supply characteristics as well as the chosen controls; the performance of any pricing/control policy is determined by this equilibrium.

All the above features (limited supply, demand heterogeneity, supply externalities and fast operational timescales) are well captured by \emph{closed queueing network models}~\cite{IntroStochNetworks,kelly2011}, which are thus widely adopted in the literature on shared vehicle systems~\cite{george2012stochastic,Waserhole2014,zhang2016control,bimpikis2016spatial,kanoria2019near,balseiro2020dynamic}. 
These models use a Markov chain to track the number of vehicles across locations. 
Each location experiences a stream of arriving customers, who engage available vehicles and take them to their desired destination.
Increasing the price for a ride between a pair of stations decreases the number of customers willing to take that ride, which over time affects the distribution of vehicles across all stations.
Even though such models can be well-calibrated, based on demand rates and price elasticities estimated from historical data, the problem of designing good pricing/control policies under such a model is complex due to a combination of high dimensionality and intrinsic non-concavity of the optimization problem (see Section~\ref{ssec:queuemodel}). 
Consequently, previous work has focused only on a narrow set of objectives (typically, weighted throughput) and is largely based on heuristics or simulation/numerical techniques, with few finite-system guarantees (see Section \ref{ssec:relwork} for a discussion). 
In contrast, finite guarantees for settings with travel-times, more complex objectives or constraints, were missing from the literature. 

In this context, we develop a \emph{rigorous approximation framework} for vehicle-sharing systems, providing a unified approach for a wide range of controls (pricing, matching, rebalancing), objective functions (throughput, revenue, welfare), and system constraints (travel-times, welfare benchmarks, posted-price constraints, limited empty-car movement). In all settings, we obtain parameterized performance guarantees, which improve with the number of vehicles in the system, and are near-optimal in the parameter regimes of real systems.
Further, we recover as special cases existing approximate-optimal policies for limited settings \cite{Waserhole2014} and asymptotic-optimality results ~\cite{george2012stochastic, braverman2016empty}. For the latter our guarantees provide an elementary proof of optimality for the so-called large-market scaling without necessitating the derivation of the associated fluid limits.
Our approach leverages techniques from convex optimization and approximation algorithms, and combines these with an infinite projection and pullback technique. Given the widespread use of closed queueing models for a variety of other applications, this technique may prove useful in other areas as well.

\subsection{Outline of our Contributions}
\label{ssec:outline}

We model a shared vehicle system with $m$ vehicles and $n$ stations as a continuous-time finite-state Markov chain that tracks the number of vehicles (units) at each station (node), and use this to study a variety of pricing and control problems. 
In Section \ref{sec:pricing} we illustrate our methodology for the vanilla setting of maximizing a single objective (e.g., throughput, revenue, or welfare) using prices, without travel times. In Sections~\ref{ssec:delay} and~\ref{sec:applications} we apply our framework to a broad spectrum of controls and constraints. 

A brief description of the vanilla model is as follows (see Section \ref{sec:model} for details): Each station in the system observes a Poisson arrival process of customers. Arriving customers draw a ride-valuation and destination from some known distribution. 
Upon arrival at a station, the customer is quoted a price and one of three scenarios occurs:\\ i) The customer is not willing to pay the price, i.e. the price exceeds her value.\\
ii) The customer is willing to pay the price but no unit is available at the node.\\
iii) The customer is willing to pay the price, and a vehicle is available. \\
In the first two cases, the customer immediately exits the system at her source station; in the final case, she engages a vehicle to take her to her destination station, and then leaves the system. Thus, in the first two cases, the system state remains unchanged; in the third, the number of vehicles at the origin is decremented by one, and at the destination is incremented by one, instantaneously. This describes the basic dynamics under which we aim to maximize the long-run average performance.
The prices can depend on the current state of the system.
Consequently, the resulting optimization problem is high-dimensional, and moreover, is not convex (see Section \ref{ssec:queuemodel}).

In Section~\ref{sec:pricing} we introduce our approximation framework for the vanilla setting described above. Our framework is based on a three-step analysis of a convex relaxation, which we refer to as the \emph{elevated flow relaxation} (see Section \ref{ssec:technical_contributions}). 
The relaxation is based on a concave point-wise upper bound which we call \emph{elevated objective}, combined with flow-conservation constraints to capture the network externalities.
The solution to the relaxation provides a state-independent pricing policy, which we prove is within a factor of $1+(n-1)/m$ with respect to the optimal state-dependent policy. This offers an alternate proof for the main result of Waserhole and Jost \cite{Waserhole2014}.

In Section \ref{ssec:delay} we relax the assumption that rides occur instantaneously to incorporate travel times between stations. We show how our approach provides guarantees that smoothly transition as travel-times increase, from a $O(1/m)$ gap with respect to the elevated flow relaxation, to a~$O(1/\sqrt{m})$ gap in so-called `heavy-traffic' regimes.
Our guarantees help quantify the parameters that control the scaling behavior in such systems. We note that these are the first (and only) finite-system results in the presence of travel-times.

In Section~\ref{sec:applications} we demonstrate the versatility of our framework by applying it to different settings:
\begin{itemize}
\item In Section \ref{ssec:rerouting} we study two other rebalancing controls considered in the literature, and get $1+(n-1)/m$ approximation guarantees for the respective optimization problems. In the first, units can be repositioned to a new location after ending a trip \cite{braverman2016empty}; in the second, customers can be matched to units at neighboring nodes \cite{ozkan2016dynamic}. In particular, we obtain an elementary proof of the large-market optimality of our algorithms, recovering the results of \cite{braverman2016empty}, whilst providing a complete analytical treatment of their empirical scaling observations (see Section \ref{ssec:delay}).
\item In Section \ref{ssec:point_pricing}, we consider a special case of the pricing problem where the platform's price multiplier is based only on the origin of a trip, and not on the destination. This model is motivated by contemporary schemes like Uber's early surge pricing as well as Lime's and Lyft's scooter operations. We show that in this case the optimization problem collapses to a one-dimensional concave maximization, allowing us to incorporate constraints that limit the permissible prices to a small set, mirroring real regulatory and operational constraints, e.g., caps discussed by \cite{zhong2020queueing}.
\item In Section \ref{ssec:multi_objective} we turn our attention to more complicated constraints such as Ramsey pricing~\cite{ramsey1927contribution}, where the goal is to maximize revenue subject to a lower bound on social welfare. 
For many shared vehicle systems this is the most relevant objective, since they are operated by private companies in close partnership with city governments. For instance, the Citi Bike system in New York City is run by Lyft under service-level agreements with the NYC Department of Transportation. Note that the complementary problem (maximizing welfare subject to revenue constraints) is of interest when such systems are managed by non-profit organizations and is considered in other paradigms such as the FCC spectrum auction \cite{milgrom}. 
In this context we demonstrate how our approach can be used to obtain a $(\gamma,\gamma)$ bicriteria approximation guarantee with $\gamma = (1+\frac{n-1}{m})$. Our results also extend to constraints that limit \emph{empty miles} driven for repositioning; this may be particularly relevant in the context of city governments in the future regulating congestion due to ``zombie cars'' (autonomous vehicles driving empty to reposition)~\cite{zombiecars}. 
\end{itemize}
Overall, our results recover and unify existing results in this area, and provide a general framework for deriving approximation algorithms for many other settings. 
Though we treat them separately, the settings above can be arbitrarily combined depending on the situation, e.g., allowing restricted point-pricing with repositioning in a scooter-sharing system, or only matching and repositioning in a ride-sharing system in which regulations forbid dynamic pricing. Moreover, the guarantees we obtain are close to 1 for realistic system parameters. For instance, for the parameters ($m=10000, n=600$) of New York City's Citi Bike system, we obtain an approximation ratio of $1.06$.

\subsection{Technical Contributions}\label{ssec:technical_contributions}  
We now briefly highlight the technical ingredients of our three-step approximation framework (see Section~\ref{sec:pricing} for a detailed description).
\begin{enumerate}
\item First, we derive an efficiently computable upper bound for the performance of any control policy.  
In particular, we observe that the solution to the elevated flow relaxation program is an upper bound on the performance of any state-dependent policy in the $m$-unit system, and hence, an upper bound for the optimal state-dependent policy in the $m$-unit system. The main idea here is to construct bounds which encode essential `conservation laws' of the system, yet are easily computable.
\item Next, we show that for a particular class of control policies in an appropriate ``large-population'' limit, the achievable objective values match the set of achievable upper bounds.
In particular, we consider a restricted subset of \emph{state-independent pricing policies}, under which the resulting Markov chain has a product-form stationary measure (in particular, corresponding to a \emph{closed BCMP network}). 
We prove that as the number of units grows to infinity, the set of `flows' achievable under such policies exactly matches the elevated flow relaxation polytope, and moreover for these policies, the elevated objective collapses to the original objective. 
\item Finally, we use the above restricted subset of control policies to derive an approximation guarantee for the performance of the policy provided by our upper bound in the finite-population setting.
Our bounds are based on showing, for all our settings, that the ratio of the performance in the $m$-unit system and the performance in the infinite-supply setting is equal to the maximum availability (steady-state probability of a vehicle being available) across all stations, and bounding that availability by $1-\frac{(n-1)}{m+n-1}$ (in the absence of travel-times, $1-O\left(\frac{1}{\sqrt{m}}\right)$ with travel-times).  
Though the intuition behind this pullback step can be observed via stochastic coupling arguments (in Appendix~\ref{app:zahorjan_proof} we include a proof based on coupling arguments from~\cite{zahorjan1982balanced}), we provide a purely combinatorial proof based on constructing a weighted bipartite graph relating the state spaces of the $m$- and $(m-1)$-unit systems. Our proof, which is based on characterizing the value of an additional unit, may be of independent interest; in particular, it is similar in spirit to the celebrated Bulow-Klemperer theorem in auction theory~\cite{bulow2009sellers} and resource augmentation arguments in competitive analysis for caching~\cite{sleator1985amortized}.
\end{enumerate}
While the first two steps are adaptations of common approaches in the revenue management literature, it is much rarer that the last can be successfully executed. In our case, this is enabled by structural results for BCMP networks, and the recognition that for all of the (state-independent) policies we design the resulting stochastic systems continue to be BCMP networks.

An interesting feature of our technique is that it demonstrates the optimality of our control policies in the so-called ``large market'' regimes~\cite{george2012stochastic,braverman2016empty}, where both number of vehicles and arrival rates of customers jointly scale to infinity. Compared to other works, however, our proof is elementary in that we do not need to characterize the limiting processes.

\subsection{Related Work}
\label{ssec:relwork}
There is a large literature on open and closed queueing network models, building on seminal work of Jackson \cite{jackson1963}, Gordon and Newell ~\cite{gordon1967closed}, and Baskett et al. \cite{baskett1975}; the books by Kelly~\cite{kelly2011} and Serfozo~\cite{IntroStochNetworks} provide an excellent summary. 
Optimal resource allocation in open queueing networks also has a long history, going back to the work of Whittle~\cite{Whittle1985}. 
However, there is much less work for closed networks, in part due to the presence of a normalization constant for which there is no closed-form (though it is computable in $O(nm)$ time via iterative techniques~\cite{Buzen,MVA}). Most existing work on optimizing closed queueing networks use heuristics, with limited or no inite guarantees. In contrast, our work focuses on obtaining algorithms with provable guarantees for a wide range of problems.

Three popular approaches for closed queueing network control in the literature are: $(i)$ using open queueing network approximations, $(ii)$ heuristically imposing a `fairness' property, which we refer to as the balanced demand constraint (see Section \ref{ssec:dcirc}), and $(iii)$ characterizing the fluid limits of closed queueing networks, and obtaining solutions that are optimal in these scaling regimes. 
We now briefly describe each approach. 

The first approach was formalized by Whitt \cite{whitt1984}, via the fixed-population-mean (FPM) method, where exogenous arrival rates are chosen to ensure the mean population is $m$. It has since been used in many applications; for example, Brooks et al. \cite{Brooks2013} use it to derive policies for matching debris removal vehicles to routes following natural disasters. Performance guarantees however are available only in restricted settings.

The second line of work is based on heuristics that enforce the balanced demand property (variously referred to as the demand rebalancing, the fairness, or the bottleneck property). In transportation settings, George et al. used these to optimize weighted throughput~\cite{george2012stochastic}, Zhang et al. to minimize rebalancing costs~\cite{zhang2016control}. Most works typically only provide asymptotic guarantees~\cite{george2012, george2011fleet}.

In concurrent work, Braverman et al. \cite{braverman2016empty} characterized appropriate fluid (or large-market) limits for closed queueing networks, and used it to study the operations of ride-sharing systems. In contrast to our work, which focuses on optimizing a given finite-$m$ system, they consider a regime where $m$ and passenger arrival-rates together scale to $\infty$, and characterize the optimal limiting policy for re-directing idle drivers to under-served locations. Ozkan and Ward~\cite{ozkan2016dynamic} also study asymptotically optimal policies for assigning customers to nearby drivers, but in an open and time-varying system.  Our extensions to settings beyond pricing (see Section \ref{ssec:rerouting}) are inspired by these works; in particular, we show that similar scaling results can be derived within our framework. Moreover, our work provides guarantees for the resulting policies in the finite case (i.e., in the pre-limit), against a more general class of state-dependent policies, and for combinations of different controls, e.g., simultaneous driver re-direction and pricing.

The closest work to ours is that of Waserhole and Jost \cite{Waserhole2014}, who provide a pricing policy for maximizing throughput in closed queueing networks, with the same approximation ratio as ours. 
They do this via a different argument, wherein they observe that if one maximizes throughput while ignoring unit-availability at nodes, then the resulting Markov chain is doubly stochastic, and hence has a uniform distribution for any number of units (this was also noted earlier by Whitt~\cite{whitt1984}). 
A simple counting argument shows that the probability of unit-availability at any node is~$m/(m+n-1)$, and hence the resulting throughput is within a $m/(m+n-1)$ factor of the optimum. This argument is finely tuned to throughput maximization via pricing, and no travel-times, and it is unclear if it extends more generally (see end of Section \ref{sec:pricing} for more details). Our work essentially re-derives this result as part of a more general framework, which then allows us to accommodate more general controls, objectives and constraints as well. 

The use of LP relaxations, as well as other ADP approaches, have a long history of use in finite-horizon control problems~\cite{adelman2007price,hampshire2009dynamic}; however these methods typically either do not scale well and/or lack provable guarantees for infinite-horizon settings. 
The work of Levi and Radovanovic~\cite{levi2010provably} is a notable exception to this in that it uses a similar approach to ours for a class of revenue management problems: they construct a relaxation to the optimal solution, and then analyze the stationary distribution under a resulting policy to prove a finite approximation guarantee versus the relaxation.

The last few years have also seen a great number of results that explicitly focus on pricing in the presence of rational drivers. Rationality tends to play a role either in the drivers' decision to join the platform or in where to drive. For the former,  Banerjee et al. \cite{banerjee2015pricing} prove that, in appropriate fluid regimes, state-independent prices are optimal, even though state-dependent ones are more robust. Chen and Hu \cite{chen2017pricing} as well as Hu and Zhou \cite{hu2017price} have obtained similar results on the optimality of state-independent prices under appropriate limits. Castillo et al. find that in keeping the available supply in a ride-sharing system high (and pick-up times consequently low) surge prices increase both social welfare and revenue \cite{castillo2017surge}. More recently authors have focused on the extent to which dynamic pricing influences where, within a region, drivers drive. Theoretical results of this flavor include those by Bimpikis et al. \cite{bimpikis2016spatial}, Castro et al. \cite{castro2018surge}, Afeche et al. \cite{afeche2018ride}, Ozkan \cite{ozkan2018joint}, and Ma et al. \cite{ma2018spatio}; these models differ from ours in both a technical and a practical sense. From a technical perspective, the main difference lies in the fact that these models tend to assume infinitesimal drivers whereas our main contribution lies in guarantees for finite settings. The main practical difference lies in the very different model of driver repositioning: our work (and that of Braverman et al. \cite{braverman2016empty}) assume that the platform offers drivers a bonus to reposition to different locations; drivers then either accept to reposition and otherwise stay put; this closely resembles operational modes like Lyft's Personalized Power Zones \cite{LyftPPZ, hao2020eppz} and Uber's new surge \cite{ubersnewsurge}. In contrast, most of the existing work assumes that the driver's payment is a fixed fraction of the passenger's payment; drivers are thus incentivized to drive where passengers are charged higher prices --- these more closely resemble the commission-based modus operandi that Uber and Lyft used to employ.

After the initial manuscript of our work, there has been a sequence of follow-up papers that directly extend our models and results. 
S\'ejourn\'e et al. \cite{sejourne2018price} study the effect of market fragmentation on the cost of rebalancing in Section~\ref{ssec:rerouting}. 
Qian et al.~\cite{banerjee2018value} identify the fastest possible rate of convergence for the matching problem in Section~\ref{ssec:rerouting}. Kanoria and Qian~\cite{kanoria2019near} prove transient guarantees for the pricing and the matching settings we consider, without requiring distributional information. Besbes et al.~\cite{besbes2019static,besbes2020pricing} study a pricing problem that is closely related to the special case in which $n=1$ and ~$m$ is small, whereas we are interested in network effects that only arise with $n>1$, and focus on ~$m$ being large. Balseiro et al.~\cite{balseiro2020dynamic} obtain results for a special case of our pricing model, with the key distinctions that they (i) consider the special case of a hub-and-spoke network (whereas we effectively study general networks), and (ii) consider an asymptotic regime in which the number of locations grows large (whereas we treat it as constant).

\section{Preliminaries}
\label{sec:model}

In this section, we first formally define our model of shared vehicle systems and formulate the problem with respect to the pricing control. The other controls follow similar formulations; to streamline the presentation, we introduce them in the corresponding sections. To capture the complex network externalities of the system, we define a probabilistic model of customer arrivals, which we analyze in steady state. Subsequently, we introduce known results from the queuing literature that provide the technical background upon which our analysis relies. 
Finally, we present an example that shows that even in the restricted sets of pricing policies that are independent of the configuration of vehicles across the system, the optimization problems we consider are non-convex.

\subsection{Basic Setting}
\label{ssec:basicsetting}

We consider a system with $m$ units (corresponding to vehicles) and $n$ nodes (corresponding to stations). Customers traveling between nodes $i$ and $j$ arrive at node $i$ according to a Poisson process of rate $\phi_{ij}$. Each customer traveling from $i$ to $j$ has a value drawn independently from a distribution $F_{ij}(\cdot)$. We assume that $F_{ij}$ has a density and that each customer has a positive value with some probability, i.e. $F_{ij}(0)<1$.
Upon arrival at $i$, a customer is quoted a price $p_{ij}$, and engages a unit to travel to $j$ if her value exceeds this price, i.e. with probability $1-F_{ij}(p_{ij})$, and at least one unit is available at node $i$; else she leaves the system immediately (i.e., we consider a loss system). If a ride occurs, the unit moves to $j$ and the customer leaves the system thereafter.

As is common with pricing, the related optimization problems are often more easily framed in terms of the \emph{inverse demand} (or \emph{quantile}) function associated with the user as $q_{ij} = 1-F_{ij}(p_{ij})$, i.e.,~$q_{ij}$ denotes the fraction of customers with origin $i$ and destination $j$ who accept price $p_{ij}$. For ease of presentation we assume that the density of $F_{ij}$
is positive everywhere in its domain (for which we only assume that it is contiguous and intersects with $(0,\infty)$), implying that there is a 1-1 mapping between prices and quantiles. As $F_{ij}$ is therefore invertible, we can write $p_{ij} = F_{ij}^{-1}(1-q_{ij})$. This allows us to abuse notation throughout the paper by using prices and quantiles interchangeably as our decision variables. Further, we define $F_{ij}(\infty)=1$, that is, we assume we can set a price high enough so that an arbitrarily small (or even 0) fraction of customers is willing to pay it.

A continuous-time Markov chain tracks the number of units across nodes. At time $t\geq 0$, the \emph{state} of the Markov chain $\X(t)=(X_1(t),\ldots,X_n(t))$ contains the number of units $X_i(t)$ present at each node $i$. For ease of presentation, we first assume that rides between nodes occur without delay. In the context of our model, this translates into an instantaneous state transition from $\X$ to $\X-e_i+e_j$ when a customer engages a unit to travel from $i$ to $j$ (where $e_i$ denotes the $i$th canonical unit vector).  The state space of the system is denoted by
$\Snm=\left\{(x_1,x_2,\ldots,x_n)\in\mathbb{N}_0^n|\sum_ix_i=m\right\}$. Throughout the paper we use $\X(t), X_{i}(t)$ to indicate random variables, and $\x,x_{i}$ to denote specific elements of the state space. Note that the state-space is finite; moreover, $|\Snm| = \binom{m+n-1}{n-1}= \Omega(m^n)$ (every state corresponds to a placement of $n-1$ stripes in between $m$ circles with the number of units at node $k$ corresponding to the number of circles  in between the $k-1$st stripe and the $k$th stripe).
Since our focus is on the long-run average performance, i.e. system performance under the steady state of the Markov chain, we henceforth suppress the dependence on $t$ for ease of notation.

\subsection{Pricing Policies and Objectives}
\label{ssec:objective}

We consider pricing policies that select point-to-point prices $p_{ij}$ as a \emph{function of the overall state $\mathbf{X}$}.
Formally, given arrival rates and demand elasticities $\{\phi_{ij},F_{ij}(\cdot)\}$, we want to design a pricing policy $\p(\cdot) = \{p_{ij}(\cdot)\}$, where each $p_{ij}:\Snm\rightarrow\R\cup\{\pm\infty\}$  maps the state to a price for a ride between $i$ and $j$. Equivalently, we want to select quantiles $\q(\cdot)=\{q_{ij}\}$ where each $q_{ij}:\Snm\rightarrow[0,1]$. 
For a fixed pricing policy $\p$ with corresponding quantiles $\q$, the \emph{effective demand stream} from $i$ to $j$ (i.e. customers traveling from node $i$ to $j$ with value exceeding $p_{ij}$) thus follows a (state-dependent) Poisson process with rate $\phi_{ij}q_{ij}(\X)$. This follows from the notion of probabilistic thinning of a Poisson process -- the rate of customers wanting to travel from $i$ to $j$ is a Poisson process of rate $\phi_{ij}$, and each customer is independently willing to pay $p_{ij}$ with probability $q_{ij}=1-F_{ij}(p_{ij})$. State-dependent prices also allow us to capture unavailability by defining $q_{ij}(\mathbf{x})=0$ if $x_i = 0$ (i.e. a customer with origin $i$ is always turned away if there are no units at that station; recall we defined $F_{ij}(\infty)=1$).
Thus, a pricing policy $\p$, along with arrival rates and demand elasticities $\{\phi_{ij},F_{ij}(\cdot)\}$, determines the transitions of the Markov chain. Note that this is a finite-state Markov chain, and hence, it has a  steady-state distribution $\pi(\cdot)$ with $\pi(\x)\geq 0\,\forall\,\x\in\Snm$ and $\sum_{\x\in\Snm}\pi(\x)=1$. For notational simplicity, we treat $\pi(\cdot)$ as unique; for all our policies we prove this in Appendix \ref{sec:irreducibility}; our upper bounds are all with respect to the best-possible steady-state distribution. Further, we allow states $\x\not\in\mathcal{S}_{n,m}$ in the domain of $\pi(\cdot)$ and set $\pi(\x):=0$ for such~$\x$. Then the following balance equation must hold for each state $\x\in\mathcal{S}_{n,m}$:
\begin{equation}
\label{eq:balance_equation}
\pi(\x)\sum_{i,j}\phi_{ij}q_{ij}(\x) = \sum_{i,j}\pi(\x+e_i-e_j)\phi_{ij}q_{ij}(\x+e_i-e_j)
\end{equation}

Our goal is to design a pricing policy $\p$ to maximize the steady-state performance under various objectives. In particular, we consider objective functions that decompose into per-ride reward functions $I_{ij}:\R\rightarrow\R$, which correspond to the reward obtained from a customer engaging a ride between stations $i$ and $j$ at price $p_{ij}$. The per-ride rewards corresponding to the three canonical objective functions are:
\begin{itemize}
\item \emph{Throughput}: the total rate of rides in the system; for this, we set $I^T_{ij}(p_{ij})=1$.
\item \emph{Social welfare}: 
the per-ride contribution to welfare, i.e., the total value of all passengers served, is given by $I^W_{ij}(p_{ij})=\mathbb{E}_{V\sim F_{ij}}\brk*{V|V\geq p_{ij}}$.
\item \emph{Revenue}: to find the system's revenue rate, we can set $I^R_{ij}(p_{ij})=p_{ij}$.
\end{itemize}
We abuse notation to define $I_{ij}(q_{ij}) \triangleq I_{ij}(F_{ij}^{-1}(1-q_{ij}))$ as a function of the quantile instead of the price.
We also define the \emph{reward curves} $R_{ij}(q_{ij}) \defeq q_{ij}\cdot I_{ij}(q_{ij})$ (analogous to the notion of revenue curves; see~\cite{hartline2013}). 
Our results require the technical condition that $R_{ij}(q_{ij})$ are concave in $q_{ij}$, which implies that $I_{ij}(q_{ij})$ are non-increasing in $q$ (equivalently $I_{ij}(p_{ij})$ are non-decreasing in $p_{ij}$). We note that this assumption holds for throughput and welfare under all considered distributions; for revenue, distributions fulfilling the condition are referred to as regular distributions \cite{hartline2013}. For completeness, we prove these observations in Appendix~\ref{sec:concave_reward_curves}.

For a given objective, our aim is to select a pricing policy $\p$, equivalently quantiles $\q$, that maximizes the steady-state rate of reward accumulation, given by
\begin{equation}
\label{eq:objgeneral} 
\obj_m(\q)=\sum_{\mathbf{x}\in\Snm}\pi(\mathbf{x})\cdot\bigg(
\sum_{i,j} \phi_{ij}\cdot q_{ij}(\mathbf{x})\cdot I_{ij}\big(q_{ij}(\mathbf{x})\big)\bigg) =
\sum_{\mathbf{x}\in\Snm}\pi(\mathbf{x})\cdot\bigg(
\sum_{i,j} \phi_{ij}\cdot R_{ij}\big(q_{ij}(\mathbf{x})\big)\bigg). 
\end{equation}
Intuitively, Equation \eqref{eq:objgeneral} captures that at any node $i$, customers destined for $j$ arrive via a Poisson process with rate $\phi_{ij}$, and find the system in state $\mathbf{x}\in\Snm$ with probability $\pi(\x)$. They are then quoted a price $p_{ij}(\x)$ (corresponding to quantile $q_{ij}(\x)$), and engage a ride with probability $q_{ij}(\x)$. The resulting ride then contributes in expectation $I_{ij}(q_{ij}(\x))$ to the objective function. Recall that unavailability of units is captured by our assumption that $q_{ij}(\x) = 0$ whenever $x_{i}=0$. Notice that the component of Equation \eqref{eq:objgeneral} that captures the flow of units traveling from $i$ to $j$ can be written~as
\begin{equation}\label{eq:state_dependent_flow}
f_{ij,m}(\q)=
\sum_{\mathbf{x}\in\Snm}\pi(\mathbf{x})\cdot \phi_{ij}\cdot q_{ij}(\mathbf{x}). 
\end{equation}
We will revisit the special case of this definition for state-independent policies in the next subsection.

\subsection{State-Independent Pricing and Product-Form Queueing Networks}
\label{ssec:queuemodel}

The Markov chain described in Section \ref{ssec:basicsetting} has the structure of a \emph{closed queueing network}~(see \cite{IntroStochNetworks,kelly2011}), a well-studied class of models in  applied probability (closed refers to the fact that the number of units remains constant; in open networks, units may arrive and depart from the system).
Our analysis crucially relies on some classical results from the queuing theory literature, which we review in this section. Our presentation here closely resembles that of Serfozo \cite{IntroStochNetworks}.

One particular class of pricing policies is that of \emph{state-independent} policies, wherein we set point-to-point prices $\{p_{ij}\}$ which do not react to the state of the system. 
Node availabilities under such policies have been studied extensively in a line of work by George, Squillante, and Xia \cite{george2012stochastic,george2011fleet,george2012}, who first derived the expressions we present here. With prices not reacting to the state of the system, the rate of units departing from any node $i$ at any time $t$ when $X_i(t)>0$ is a constant, independent of the state of the network.
The resulting model is a special case of a closed queueing model proposed by Gordon and Newell~\cite{gordon1967closed}.
\begin{defn}
\vspace{0.1in}
\label{def:gn} 
A \emph{Gordon-Newell network} is a continuous-time Markov chain on states $\x\in \mathcal{S}_{n,m}$, in which for any state $\x$ and any $i,j\in[n]$, the chain transitions from $\x$ to $\x-e_i+e_j$ at a rate $P_{ij}\mu_i\mathds{1}_{\{x_i(t)>0\}}$, where $\mu_i>0$ is referred to as the \emph{service rate} at node $i$, and $P\geq 0$ as the \emph{routing probabilities} satisfying $\sum_jP_{ij}=1$.
\end{defn}
In other words, if units are present at a node $i$ in state $\x$, i.e. if $x_i>0$, then departures from that node occur according to a Poisson distribution with rate $\mu_i>0$; conditioning on a departure, the destination $j$ is chosen according to state-independent routing probabilities $P_{ij}$.

The Markovian dynamics resulting from state-independent pricing policies fulfill the conditions of Gordon-Newell networks: fixing a price $p_{ij}$ (with corresponding $q_{ij}$) results in a 
Poisson process with rate $\phi_{ij}q_{ij}$ of arriving customers \emph{willing to pay price} $p_{ij}$.
These customers engage a unit only if one is available, else they leave the system. Thus, given quantiles $\q$, the time to a departure from node $i$ is distributed exponentially with rate $\mu_i=\sum_j\phi_{ij}q_{ij}$ when $X_i>0$ and with rate 0 otherwise. Further, conditioned on an arriving customer having value at least equal to the quoted price, the probability that the customer's destination is $j$, is $P_{ij}=\phi_{ij}q_{ij}/\sum_k \phi_{ik}q_{ik}$, independent of system state.

One advantage of considering state-independent policies (and drawing connections with Gordon-Newell networks) is that the resulting steady-state distribution $\crl*{\pi_{\p,m}(\x)}_{\x\in\Snm}$ can be expressed in product form, as established by the Gordon-Newell theorem. 
\begin{theorem}[Gordon-Newell Theorem~\cite{gordon1967closed}]
\vspace{0.1in}
Consider an $m$-unit $n$-node Gordon-Newell network with transition rates $\mu_i$ and routing probabilities $P_{ij}$. Let $\crl*{w_i}_{i\in[n]}$ denote the invariant distribution associated with the routing probability matrix $\crl*{P_{ij}}_{i,j\in[n]}$, and define the \emph{traffic intensity} at node $i$ as $r_i=w_i/\mu_i$. Then the stationary distribution is given by: 
\begin{equation}
\label{eq:steadystate}
   \pi(\x)=\frac{1}{G_m} \prod_{j=1}^n \prn*{r_j}^{x_j},
\end{equation}
where the Gordon-Newell normalization constant is given by $G_m=\sum_{\x\in \mathcal{S}_{n,m}}\prod_{j=1}^n\prn*{r_j}^{x_j}$.
\end{theorem}
To obtain some intuition for the form of Equation \eqref{eq:steadystate}, consider a system with just one unit. If $\mu_i$ is equal to 1 for all $i$, then $w_i$ exactly captures the stationary distribution of a simple random walk with routing matrix $P$. Further, changing $\mu_i$ at a node $i$ does not affect how frequently the unit visits $i$, but it affects how much time the unit spends at $i$ before departing again. The distribution above can now be seen to be the occupancy distribution induced by $m$ independent random walks. For a formal proof of how this stationary distribution arises in the case of the Gordon-Newell network, we refer the reader to Serfozo~\cite{IntroStochNetworks}.

The Gordon-Newell network is a special case of a more general family of queueing models, called \emph{product-form} networks, characterized by stationary distributions similar to Equation~\eqref{eq:steadystate}. In particular, given any collection of integrable functions $\rho_j:\mathbb{N} \rightarrow \mathbb{R}_+$, a queueing network is said to be product-form if its stationary distribution is given by $\pi(\x)=\frac{1}{Z_m} \prod_{j=1}^n \rho_j(x_j)$, where $Z_m$ is an appropriate normalizing constant (or \emph{partition function}). The Gordon-Newell network is a special case with $\rho_j(x) = r_j^x$ (and normalizing constant $Z_m$ denoted as $G_m$); we introduce another example, the BCMP network~\cite{baskett1975}, below when we model non-zero travel times. An important property in common to all such networks is the following characterization of their steady-state availability.
\begin{lemma}[availability in product-form networks, Eq. (5) in \cite{george2011fleet}]
\label{lem:avail}
\vspace{0.1in}
Consider any product-form queueing network with $m$ units, $n$ nodes, and functions $\{\rho_j\}_{j\in[n]}$, and any node $i\in[n]$ with $\rho_i(y) = r_i^y$. Let $A_{i,m} \triangleq \sum_{\x\in\mathcal{S}_{n,m}}\pi(\x)\mathds{1}_{\{x_i>0\}}$ be the steady-state \emph{availability} of units at node $i$ (i.e. the probability in steady-state that at least one unit is present at node $i$). Then we have
\begin{equation}
\label{eq:avail}
A_{i,m}=
\prn*{Z_{m-1}/Z_m}\cdot r_i.
\end{equation}
\end{lemma}
We can use the Gordon-Newell theorem and Lemma~\ref{lem:avail} to simplify the objective function in Equation \eqref{eq:objgeneral}.
Recall that for an $m$-unit system with state-independent policy $\p$ (with corresponding quantiles $\q$), we obtain a Gordon-Newell network with service rate $\sum_j \phi_{ij}q_{ij}$ and routing probabilities $\phi_{ij}q_{ij}/\sum_k \phi_{ik}q_{ik}$ at node $i$.
Let $\crl*{\pi(\x)}_{\x\in\Snm}$ be the corresponding steady-state distribution. 
Since $\q$ is no longer a function of the system state, we can no longer set $q_i=0$ when $X_i=0$. Instead, as in Lemma~\ref{lem:avail},
we define $A_{i,m}(\q) = \sum_{\x\in\mathcal{S}_{n,m}}\pi(\x)\mathds{1}_{\{x_i>0\}}$ as the steady-state \emph{availability} of units at node $i$; the Lemma now implies that $A_{i,m}(\q)=\prn*{G_{m-1}(\q)/G_m(\q)}\cdot r_i(\q)$. Moreover, from Equation \eqref{eq:state_dependent_flow}, the state-independent \emph{steady-state rate of units} moving from $i$ to $j$ then simplifies to $f_{ij,m}(\q) = A_{i,m}(\q)\cdot\phi_{ij}q_{ij}$.
Now, the objective in Equation \eqref{eq:objgeneral} can be written as 
\begin{equation}
\label{eq:objective}
\obj_m(\q)= \sum_{i}A_{i,m}(\q)\cdot\prn*{\sum_j \phi_{ij}q_{ij}\cdot I_{ij}(q_{ij})} = \sum_{i,j} f_{ij,m}(\q)I_{ij}(q_{ij}).
\end{equation}
For ease of notation, we omit the explicit dependence on $m$ when clear from context.

\paragraph{Non-concavity of objective function.} 
Directly optimizing the finite-unit system is non-trivial as the objective function is not concave in prices (or quantiles); we now demonstrate this in a simple network ($m=1$ and $n=3$), using throughput as the objective. Notice that in the setting with just one vehicle, there is no distinction between state-dependent and state-independent policies. 

Our example is presented in Figure \ref{fig:nonconcave}. The network comprises of three nodes ($A,B,C$); the labels on the edges show the demand rate $\phi_{ij}$ with which people wanting to move from node~$i$ to node~$j$ arrive. Instead of computing in detail the throughput obtained by different quantiles, we just provide the intuition for non-concavity by showing that (i) the throughput is small (of order $\epsilon$) when all quantiles are equal to 1, (ii) it is still almost as small when $q_{BC}=\frac{1+\epsilon}{2}$ and the other quantiles are 1, and (iii) it is large when $q_{BC}=\epsilon$ and the other quantiles are 1. This contradicts concavity since the objective at $q_{BC}=\frac{1+\epsilon}{2}$ is smaller than the mean of the objectives at $q_{BC}=1$ and $q_{BC}=\epsilon$ (keeping the other quantiles constant). In the first case, the single unit moves from $B$ to $C$ about half the times it departs from $B$; the expected time before an arrival at $C$ is then equal to $(1/\epsilon)$, so it is not serving any rides for a long time. When $q_{BC}=\frac{1+\epsilon}{2}$ the unit still moves to $C$ about one out of every three times it departs from $B$ and then spends a lot of time waiting for an arrival at $C$. However, if $q_{BC}=\epsilon$, it only moves to $C$ a fraction $\frac{\epsilon}{1+\epsilon}$ of the times it departs from $B$, canceling out the wait of length $1/\epsilon$ at $C$. In that case, the throughput is equal to $\frac{2+\epsilon}{3}$ (this also happens to be the solution our Algorithm in Section \ref{ssec:dcirc} returns). Note that this also implies non-concavity in prices (e.g., if all passengers have uniform value distributions).

\begin{figure}[!h]
\centering
\centering
\scalebox{1}{

\tikzset{every picture/.style={line width=0.75pt}} 

\begin{tikzpicture}[x=0.75pt,y=0.75pt,yscale=-1,xscale=1]

\draw   (100,128) .. controls (100,114.19) and (111.19,103) .. (125,103) .. controls (138.81,103) and (150,114.19) .. (150,128) .. controls (150,141.81) and (138.81,153) .. (125,153) .. controls (111.19,153) and (100,141.81) .. (100,128) -- cycle ;
\draw   (200,128) .. controls (200,114.19) and (211.19,103) .. (225,103) .. controls (238.81,103) and (250,114.19) .. (250,128) .. controls (250,141.81) and (238.81,153) .. (225,153) .. controls (211.19,153) and (200,141.81) .. (200,128) -- cycle ;
\draw   (300,128) .. controls (300,114.19) and (311.19,103) .. (325,103) .. controls (338.81,103) and (350,114.19) .. (350,128) .. controls (350,141.81) and (338.81,153) .. (325,153) .. controls (311.19,153) and (300,141.81) .. (300,128) -- cycle ;
\draw    (125,103) .. controls (164.6,73.3) and (180.02,74.64) .. (223.67,102.16) ;
\draw [shift={(225,103)}, rotate = 212.39] [color={rgb, 255:red, 0; green, 0; blue, 0 }  ][line width=0.75]    (10.93,-3.29) .. controls (6.95,-1.4) and (3.31,-0.3) .. (0,0) .. controls (3.31,0.3) and (6.95,1.4) .. (10.93,3.29)   ;
\draw    (225,103) .. controls (264.6,73.3) and (280.02,74.64) .. (323.67,102.16) ;
\draw [shift={(325,103)}, rotate = 212.39] [color={rgb, 255:red, 0; green, 0; blue, 0 }  ][line width=0.75]    (10.93,-3.29) .. controls (6.95,-1.4) and (3.31,-0.3) .. (0,0) .. controls (3.31,0.3) and (6.95,1.4) .. (10.93,3.29)   ;
\draw    (227.11,154.51) .. controls (261.54,178.7) and (291.03,179.13) .. (325,153) ;
\draw [shift={(225,153)}, rotate = 36] [color={rgb, 255:red, 0; green, 0; blue, 0 }  ][line width=0.75]    (10.93,-3.29) .. controls (6.95,-1.4) and (3.31,-0.3) .. (0,0) .. controls (3.31,0.3) and (6.95,1.4) .. (10.93,3.29)   ;
\draw    (127.11,154.51) .. controls (161.54,178.7) and (191.03,179.13) .. (225,153) ;
\draw [shift={(125,153)}, rotate = 36] [color={rgb, 255:red, 0; green, 0; blue, 0 }  ][line width=0.75]    (10.93,-3.29) .. controls (6.95,-1.4) and (3.31,-0.3) .. (0,0) .. controls (3.31,0.3) and (6.95,1.4) .. (10.93,3.29)   ;

\draw (117,114) node [anchor=north west][inner sep=0.75pt]   [align=left] {{\Large A}};
\draw (217,116) node [anchor=north west][inner sep=0.75pt]   [align=left] {{\Large B}};
\draw (317,116) node [anchor=north west][inner sep=0.75pt]   [align=left] {{\Large C}};
\draw (169,83.4) node [anchor=north west][inner sep=0.75pt]  [font=\Large]  {$1$};
\draw (269,84.4) node [anchor=north west][inner sep=0.75pt]  [font=\Large]  {$1$};
\draw (170,144.4) node [anchor=north west][inner sep=0.75pt]  [font=\Large]  {$1$};
\draw (272,144.4) node [anchor=north west][inner sep=0.75pt]  [font=\Large]  {$\varepsilon $};
\end{tikzpicture}
}
\caption{Example for non-concavity of throughput for finite units: The network comprises of three nodes ($A,B,C$); the labels on the edges show the demand rate $\phi_{ij}$ with which people wanting to move from node~$i$ to node~$j$ arrive. In this setting, if we vary the price between $B$ and $C$, while setting the admitted quantiles to $1$ for all other pairs, then the resulting throughput is non-concave.}
\label{fig:nonconcave}
\end{figure}

\subsection{Model of Travel Times}\label{ssec:model_delays}
We now explain how to extend this model from a system in which rides occur instantaneously to one in which travel-times are non-zero. A standard way to model travel-times is to assume that each unit takes an i.i.d. random time to travel from node~$i$ to~$j$ (see, e.g., \cite{george2012, banerjee2015pricing}).
Formally, we expand the network state to $\X=\crl*{X_i(t),X_{ij}(t)}$, where \emph{node queues} $X_i(t)$ track the number of available units at node $i$, and \emph{link queues} $X_{ij}(t)$ track the number of units in transition between nodes $i$ and~$j$. 
When a customer engages a unit to travel from $i$ to $j$, the state changes to $\X-e_i+e_{ij}$ (i.e.,~$X_i\rightarrow X_i - 1$ and  $X_{ij}\rightarrow X_{ij}+1$).
The unit remains in transit for an i.i.d. random time with mean $\tau_{ij}$. 
When the unit reaches its destination, the state changes to $\X-e_{ij}+e_{j}$.
Finally, we assume that pricing policies and passenger-side dynamics remain the same as before; in particular, we assume that the demand characteristics $\crl*{\phi_{ij},F_{ij}}$ and reward-functions $\crl*{I_{ij}}$ are independent of the actual transit times (dependence on average transit times $\tau_{ij}$ can be embedded in the functions).

For notational convenience, we assume that the travel time is exponential; however, our analysis holds for any general travel-time distribution (as long as each unit experiences an independent travel time). This follows from standard insensitivity properties of $M/GI/\infty$ queues, in particular, the facts that the stationary distribution depends only on the mean service-time, and that the departure process is a Poisson process~\cite{newell1966m,foley1986stationary}.

The system described above is a generalization of the Gordon-Newell network referred to as a \emph{BCMP network} (introduced by \cite{baskett1975}; see \cite{IntroStochNetworks}, Section 3.3; also see \cite{zhang2016control} for the use of such a model for vehicle sharing). It is also a special case of a closed migration process; our presentation here follows Kelly and Yudovina \cite{kelly2014stochastic} (Chapter 2).

\begin{defn}\label{def:migration}\vspace{0.1in}
A closed migration process on states $\mathcal{S}_{n^2,m}$ is a continuous-time Markov chain in which transitions from state $\X$ to state $\X-e_i+e_j$ occur at rate $P_{ij}\mu_i(X_i)$ when $X_i>0$ and at rate 0 otherwise. The $P_{ij}$ again form routing probabilities with $\sum_kP_{ik}=1, P_{ij}\geq 0\;\forall i,j$. Notice that $\mu_i(X_i)$ is a function of $X_i$ only, whereas $P_{ij}$ are independent of the state altogether. 
\end{defn}

Given quantiles $\q$, the above-described process is a closed migration process with $P_{i,ij}=\phi_{ij}q_{ij}/\sum_k\phi_{ik}q_{ik}$ and $P_{ij,j}=1$ for every $i$ and $ j$. Further, the service rate $\mu_i(X_i)=\sum_k \phi_{ik}q_{ik}$ when $X_i>0$ for node queues and $\mu_{ij}(X_{ij})=X_{ij}/\tau_{ij}$  for link queues. Intuitively, the latter captures the idea that each of the $X_{ij}$ units has an exponential rate of $1/\tau_{ij}$ and therefore the rate until the first is removed from the link queue is $X_{ij}/\tau_{ij}$. The stationary distribution can then be obtained as follows.

\begin{theorem}[Theorem 2.4 in ~\cite{kelly2014stochastic}]\vspace{0.1in} For a closed migration process as described in Definition \ref{def:migration}, let $\crl*{w_i}_{i\in[n^2]}$
denote the invariant
distribution associated with the routing probability matrix $\crl*{P_{ij}}_{i,j\in[n]}$. 
Then the equilibrium distribution for a closed migration process is
$$
\pi(\x) = \frac{1}{Z_m} \prod_{i=1}^{n^2} \frac{w_i^{x_i}}{\prod_{y=1}^{x_i}\mu_i(y)},
$$
where $Z_m= \sum_{\x} \prod_{i=1}^{n^2} \frac{w_i^{x_i}}{\prod_{y=1}^{x_i}\mu_i(y)}$ is a normalizing constant.
\end{theorem}

This implies for our setting, with $\w$ denoting again the invariant distribution of the routing matrix.

\begin{equation}
\label{eq:steadystateBCMP}
\pi_{\x,m}(\q)=
\frac{1}{Z_m(\q)}\brk*{\prod_{i\in[n]} \prn*{\frac{w_i(\q)}{\sum_{k}\phi_{ik}q_{ik}}}^{x_i}}
\brk*{\prod_{i,j\in [n]^2}\frac{\prn*{\tau_{ij}w_{ij}(\q)}^{x_{ij}}}{x_{ij}! }}.
\end{equation}

We remark that in comparison to the invariant distribution $\w^I$ when rides occur instantaneously,~$\w^D$ with delays would be $w_i^D=w_i^I/2$ for node queues and $w_{ij}^D=\frac{w_i^D\phi_{ij}q_{ij}}{\sum_k\phi_{ik}q_{ik}}$ for link queues; this is because units alternate between being in node queues and being in link queues.

\paragraph{The infinite-unit limit:}
The stationary distributions described above for state-independent pricing policies, either with or without travel-times, hold for any finite $m$.
Moreover, it can be shown (see Theorem 3.18 in \cite{IntroStochNetworks}) that as $m\rightarrow\infty$, the stationary distribution of the $m$-unit system (as specified in Equation \eqref{eq:steadystate}) converges in distribution to a limiting distribution. The exact form of this limiting system is well understood (see Section 3.7 in \cite{IntroStochNetworks}, Theorem 2.2.1), but not consequential for our work; one important fact to note about the limiting distribution is the existence of a node with availability $1$. This essentially captures the fact that in an infinite-unit system, at least one node must have an infinite number of units. While an analytical proof of this result can be found in Section 3.7 of \cite{IntroStochNetworks}, Lemma \ref{lem:finite} (without travel-times) and Lemma \ref{lem:Am_w_delay} (with travel-times) can be interpreted as providing a combinatorial one. This fact, combined with Equation \eqref{eq:avail}, gives the following proposition.
\begin{prop}[Theorem 2.2.1 in \cite{george2012stochastic} and Theorem 1 in \cite{george2011fleet}]
\label{prop:infinite}\vspace{0.1in}
Recall that given $\q=\{q_{ij}\}$, the quantities $w_i(\q)$ and $r_i(\q)$ are independent of~$m$. Then, given a policy with quantiles $\q$, in the infinite-unit limit, the steady-state availability of each node $i$ is given by $r_i(\q)/\max_j r_j(\q)$.
\end{prop}

To see this note that for any $m$, the maximum availability node $i^{\star}$ obeys $r_{i^{\star}}(\q) = \max_j r_j(\q)$. Moreover, as $m\to\infty$,  $\prn*{G_{m-1}(\q)/G_m(\q)}\cdot r_{i^\star}(\q)\to1$; we thus obtain that $G_{m-1}(\q)/G_m(\q)\to 1/r_i^\star(\q)$, which in turn implies that  $A_{j,m}(\q)=\prn*{G_{m-1}(\q)/G_m(\q)}\cdot r_j(\q) \to r_j/r_{i^\star}(\q)$.

\section{Our Approximation Framework}
\label{sec:pricing}
In this section, we introduce our approximation framework which we build on throughout the paper. To convey the main ideas of our analysis, we first apply it to the most vanilla setting: using pricing as a control to maximize throughput when there are no travel times. This is the only setting in which a finite guarantee exists in the literature (see~Waserhole and Jost~\cite{Waserhole2014}); we reprove this result (and extend it to other objectives) with a slightly more complicated technique in order to illustrate the ingredients of our framework. In Sections \ref{ssec:delay} and~\ref{sec:applications},  we show how this framework yields guarantees for a wide range of other settings including travel times, controls beyond pricing, and constrained pricing settings.

Even for the vanilla setting, there are two technical hurdles. First, we want to compare against state-dependent pricing policies in which the number of potentially distinct prices can be exponential in the number of units. Second, even if we restrict to state-independent pricing policies, the resulting problem is non-convex in the resulting quantiles as discussed in Section \ref{ssec:queuemodel}. We circumvent these hurdles by working with a convex relaxation that restricts the policy search; for the throughput objective, the relaxation is the same as in \cite{Waserhole2014}. Our approximation framework first bounds the optimal objective by the objective of the convex relaxation and then relates the convex relaxation to a system with infinitely many units. Finally, the key analytical ingredient is to bound the objective of any policy in the finite-unit system with its performance in the infinite-unit system.

\subsection{Elevated Objective Function}\label{ssec:elev_obj}
We restrict our attention to state-independent pricing policies which reduces the output of the program from an exponential to a quadratic number of distinct prices (one for each source-destination pair). Recall from 
Equation \eqref{eq:objective} that the steady-state objective for state-independent policies can be written as
$$
\obj_m(\q)=\sum_{i,j}f_{ij,m}(\q)\cdot I_{ij}(q_{ij})
$$
here $f_{ij,m}(\q)=A_{i,m}(\q)\cdot \phi_{ij}q_{ij}$ are the resulting steady-state rates of units, which we also refer to as \emph{steady-state flows}. For throughput, the objective is significantly simplified as $I_{ij}^T(q_{ij})=1$.

We first need to make a distinction between \emph{theoretical quantiles} $q_{ij} = 1-F_{ij}(p_{ij})$ and \emph{effective quantiles} $\widehat{q}_{ij}=f_{ij,m}(\q)/\phi_{ij} = A_{i,m}(\q)\cdot q_{ij}$. Whilst theoretical
quantiles are in one-to-one correspondence to prices, effective quantiles incorporate thinning both due to the demand elasticity and the unavailability of units. Note that $\widehat{q}_{ij}\leq q_{ij}$; moreover, although any $q_{ij}\in[0,1]$ can be induced via an appropriate price, not all $\widehat{q}_{ij}$ can be realized by prices as effective quantiles (e.g., if $m<n$, not all effective quantiles can be equal to 1).
Since we assume that the per-ride rewards $I_{ij}(\cdot)$ are non-increasing on the quantile space, we have $I_{ij}(q_{ij})\leq I_{ij}(\widehat{q}_{ij})$; for throughput this holds with equality, since $I_{ij}^T(\cdot)=1$. 

We now define the \emph{elevated objective function} as
\begin{equation}\label{eq:elevated}
\widehat{\obj}(\q)=\sum_{i,j} \phi_{ij}q_{ij}I_{ij}(q_{ij})=\sum_{i,j} \phi_{ij} R_{ij}(q_{ij}).
\end{equation}
The elevated objective is essentially the objective of state-independent quantiles $\q$ ignoring demand thinning due to unavailability. It has two useful properties: $i)$ for all $m$ and $\q$, the elevated objective upper bounds the true objective function, i.e. $\widehat{\obj}(\q)\geq \obj_m(\q)$, and $ii)$ it is a concave function of $\q$ (since we focus on objectives corresponding to concave reward curves $R_{ij}(\cdot)$). 
\subsection{The Flow Polytope}\label{ssec:flow_polytope}
To make the elevated objective function relevant, we need to reinstate the effect of network externalities that we completely disregarded by ignoring unavailability. Therefore we impose a set of necessary (though not sufficient) properties that the steady-state rate of flows $\crl{f_{ij,m}(\q)}$ and the effective quantiles $\widehat{\q}$ must satisfy in order to be realizable. These properties form a linear polytope on these variables, which we refer to as \emph{flow polytope} as it relates to flow conservation and capacity constraints. We begin by proving that both properties are indeed necessary.

\begin{prop}[Demand bounding]\vspace{0.1in} For theoretical 
quantiles $\q$ under any state-dependent policy, the steady-state rate of flows obeys the demand bounding property $f_{ij,m}(\q)\leq \phi_{ij}$.
\end{prop}
\begin{proof}
The proof follows immediately from $q_{ij}(\cdot)\leq 1$ which implies $\sum_{\x\in\mathcal{S}_{n,m}} \pi(\x) q_{ij}(x)\leq 1$.
\end{proof}
\begin{prop}[Supply Circulation]\label{prop:supply_circ}
For theoretical quantiles $\q$ under any state-dependent policy in the finite $m$-unit system, the steady-state rate of flows obey flow conservation $\sum_k f_{ki,m}(
\q)=\sum_j f_{ij,m}(\q)$ for every $i$. 
The same holds for state-independent policies in the infinite-unit system.
\end{prop}
\begin{proof}
First consider state-dependent policies $\q$ in the finite $m$-unit system. For a given node~$i$, denote the set of states in which $i$
has $\ell$ units by $L^{i}_\ell$, i.e., we define $L^{i}_\ell=\{\x\in \mathcal{S}_{n,m}:x_{i}=\ell\}$ for $\ell\in\mathbb{Z}$; observe that  $L_{\ell}^i=\emptyset$ for $\ell<0$ and $\ell>m$. We also define $L^{i}_{\geq\ell}=L^{i}_\ell\cup L^{i}_{\ell+1}\cup\ldots\cup L^{{i}}_m$. We now focus on the steady-state flow balance conditions for the cut defined by $L^{i}_{\geq\ell}$. We first sum Equation \eqref{eq:balance_equation} across all states $\x\in L^{i}_{\geq\ell}$ to obtain
$$
\sum_{\x\in L^{i}_{\geq\ell}}\pi(\x)\sum_{k,j}\phi_{kj}q_{kj}(\x) = \sum_{\x\in L^{i}_{\geq\ell}}\sum_{k,j}\pi(\x+e_{k}-e_j)\phi_{kj}q_{kj}(\x+e_{k}-e_j).
$$
For each $(\x,k,j)$ on the LHS such that $\x-e_{k}+e_j\in L^{i}_{\geq\ell}$, we have the same summand appear on the RHS; further, for each $\x+e_{k}-e_j$ on the RHS, we have the same summand appear on the LHS if $\x+e_{k}-e_j\in L^{i}_{\geq \ell}$. Thus, we can cancel out terms on both sides. Terms remaining on the LHS are triples $(\x,k,j)$  such that $\x\in L_{\geq \ell}^i$ and $\x-e_k+e_j\not\in L_{\geq \ell}^i$;  summands remaining on the RHS are~$\x+e_{k}-e_j$ with $\x\in L_{\geq \ell}^i,\x+e_k-e_j\not\in L^i_{\geq \ell}$, i.e. on the LHS we only have terms indexed by $k=i$, and on the RHS we only have terms indexed by $j=i$. Thus, we get
$$
\sum_j \sum_{\x\in L^{i}_{\geq\ell}:\x-e_i+e_j\not\in L^i_{\geq\ell}}\pi(\x)\phi_{ij}q_{ij}(\x)=\sum_{k}\sum_{\x\in L^
{i}_{\geq\ell}:\x+e_k-e_
{i}\not\in L^
{i}_{\geq\ell}}\pi(\x+e_{k}-e_
{i})\phi_{ki}q_{ki}(\x+e_k-e_{i}) .
$$
Observe that the inner sum in the LHS is over only states in $L_{\ell}^i$ while the inner sum in the RHS is only over states in $L_{\ell-1}^i$. Interchanging the two sums and summing over all $\ell$, we
obtain the equality
\begin{align*}
\sum_{\ell=0}^m\sum_{\x\in L^{i}_\ell}\pi(\x)\sum_{j}\phi_{ij}q_{ij}(\x)
=\sum_{\ell=0}^m\sum_{\x\in L^{i}_{\ell-1}}\sum_{k}\pi(\x)\phi_{ki}q_{{ki}}(\x)
\end{align*}
which implies that $\sum_j f_{ij,m}(\q)=\sum_k f_{ki,m}(\q)$ by Equation \ref{eq:state_dependent_flow} and proves the first claim.

For the second claim, note that for state-independent policies and any $m$, we defined  $w_{i}(\q)$ to be the leading left eigenvector of $\{P(\q)\}_{i,j}$, where $P_{ij}(\q) = \phi_{ij}q_{ij}/\sum_j\phi_{ij}q_{ij}$. From this we get for all $i$:
\begin{align*}
\sum_{j}w_{j}(\q)\frac{\phi_{ji}q_{ji}}{\sum_k\phi_{jk}q_{jk}} = w_i(\q) = \frac{w_i(\q)}{\sum_k\phi_{ik}q_{ik}}\prn*{\sum_{k}\phi_{ik}q_{ik}}
\Rightarrow \sum_{j}r_{j}(\q)\phi_{ji}q_{ij} =  \sum_{k}r_i(\q)\phi_{ik}q_{ik}
\end{align*}
Multiplying both sides by $\prn*{G_{m-1}(\q)/G_m(\q)}$, and using our previous formula for the node availability (see~Equation~\ref{eq:avail}), we get $\sum_{j}A_{j,m}(\q)\phi_{ji}q_{ij} =  \sum_{k}A_{i,m}(\q)\phi_{ik}q_{ik}$. By Proposition \ref{prop:infinite} this also holds in the infinite-unit limit.
\end{proof}

Interpreting the demand bounding and the supply circulation property in terms of the \emph{effective quantiles} of state-independent policies, we find the linear constraints
$$
\widehat{q}_{ij}\in[0,1] \text{ and }
\sum_k \phi_{ki}\widehat{q}_{ki}=\sum_j \phi_{ij}\widehat{q}_{ij}.
$$
Theoretical quantiles need not fulfill the supply circulation property, but any (state-independent) quantiles induce flows (effective quantiles) that fulfill it.

\subsection{Pricing via the Elevated Flow Relaxation}
\label{ssec:dcirc}

Combining the elevated objective and the flow polytope, we obtain the \emph{elevated flow relaxation program} (see~Algorithm~\ref{alg:demand_circulation_reduction_general_pricing}). 
For the  
throughput objective, this is a linear optimization problem since both objective function and polytope are linear as $R_{ij}^T(q_{ij})=q_{ij}$. For objectives with concave reward curves, this is a convex program (recall the definition of reward curves in Section~\ref{ssec:objective}). 
The concave reward curves assumption is satisfied by the social welfare objective unconditionally and it is also satisfied by the revenue obective assumning concave revenue curves (Appendix~\ref{sec:concave_reward_curves}).
Note that the supply circulation property implied flow conservation of the effective quantiles at each node. Algorithm \ref{alg:demand_circulation_reduction_general_pricing} drops the availability term in the objective but imposes, as a constraint, the same demand bounding and flow conservation properties on the theoretical
quantiles. 
\begin{algorithm}[!h]
\caption{The Elevated Flow Relaxation Program with Pricing}
\label{alg:demand_circulation_reduction_general_pricing}
\begin{algorithmic}[1]
\REQUIRE arrival rates $\phi_{ij}$, value distributions $F_{ij}$,  reward curves~$R_{ij}$.
\STATE Find $\crl*{\widetilde{q}_{ij}}$ that solve the following relaxation:
$$
\begin{array}{rcll}
\max_\q &&  \sum_{(i,j)} \phi_{ij}R_{ij}(q_{ij}) 
&\\
\sum_{k} \phi_{ki}q_{ki}  & = & \sum_{j} \phi_{ij}q_{ij} & \forall\;i\\
q_{ij} & \in & [0,1] & \forall\;i,j.
\end{array}
$$
\STATE Output \emph{state-independent} prices $
\widetilde{p}_{ij}= F_{ij}^{-1}(1-\widetilde{q}_{ij})$ and the corresponding quantiles $\widetilde{q}_{ij}$.
\end{algorithmic}
\end{algorithm}

It is important to notice the distinction between the above constraint on theoretical
quantiles and the supply circulation property of effective quantiles. The theoretical
quantiles $\widetilde{q}_{ij}$ returned by the algorithm are \emph{constrained} to satisfy $\sum_{k}\phi_{ki}\widetilde{q}_{ki}=\sum_j\phi_{ij}\widetilde{q}_{ij}$ at every node $i$. They then give rise to effective quantiles $\widehat{\q}$ (with $\widehat{q}_{ij}=A_{i,m}(\widetilde{\q})\cdot\widetilde{q}_{ij}$) which obey the supply circulation property $\sum_k \phi_{ki}\widehat{q}_{ki}=\sum_j \phi_{ij}\widehat{q}_{ij}$ (as proven for any effective quantiles in Proposition \ref{prop:supply_circ}). 
In other words, the linear program above restricts our pricing policy to induce theoretical
quantiles which obey the flow circulation property, thereby mirroring a feature of effective quantiles.
Henceforth, for any theoretical quantile $\q$, we refer to the property $\sum_{k}\phi_{ki}q_{ki}=\sum_j\phi_{ij}q_{ij}$ as \emph{balanced demand} based on the terminology of \cite{george2012stochastic,george2011fleet} who use the term \emph{balanced network.} Note that this should be distinguished from supply circulation which is a property of the effective quantiles.

\subsection{The Approximation Guarantee}
Our analysis relies on the three following lemmas. Lemma~\ref{lem:state_dependent} shows that the solution of the elevated flow relaxation upper bounds the true objective of any state-depedendent pricing policy. Lemma~\ref{lem:infinite} then shows that the true (non-elevated) objective of the pricing policy returned by our program is equal to the solution of the program, when applied to an infinite-unit system. Finally, Lemmas~\ref{lem:max_availability} and~\ref{lem:finite} enables us to connect this solution to the true objective of the finite-unit system by showing that the true objective of any policy in the $m$-unit system is within a factor of $\frac{m}{m+n-1}$ of the objective in the infinite-unit system.

\begin{lemma}[from finite-unit state-dependent to the elevated flow relaxation]
\label{lem:state_dependent}
\vspace{0.1in} For objectives with concave reward curves $R_{ij}(\cdot)$, the  value of the objective function of the optimal state-dependent policy is upper bounded by the value of the elevated objective function of the pricing policy $\widetilde{\q}$ returned by the elevated flow relaxation program:
\begin{equation*}
\widehat{\obj}(\widetilde{\q})\geq \opt_m.
\end{equation*}
\end{lemma}\quad
\begin{proof}
For simplicity, we first focus on the throughput objective. The optimal state-dependent pricing policy $\q^\star(\cdot)$ induces a steady-state distribution $\pi^\star(\cdot)$. Based on that distribution we define, analogously to the effective quantiles of state-independent policies, the average-fraction of customers that receive service for each origin-destination pair: $\q=\sum_{\x\in S_{n,m}} \pi^{\star}(\x)\cdot \q^{\star}(\x)$. Then
\begin{equation}\label{eq:opt_obj_throughput}
\opt_m=\sum_{\x\in S_{n,m}}\pi^{\star}(\x)\sum_{i,j}\phi_{ij}q^{\star}_{ij}(\x)=\sum_{ij} \phi_{ij} q_{ij}=\widehat{\obj}(\q).
\end{equation}
We complete the proof by showing that $\widehat{\obj}(\widetilde{\q})\geq\widehat{\obj}(\q)$. To do so, we demonstrate that $\q$ is a feasible solution of the elevated flow relaxation program; this suffices as $\widetilde{\q}$ maximizes the elevated objective over the feasible set. We thus only need to show that $\q$ satisfies the demand bounding and balanced demand properties. The first holds trivially since $\q$ is a convex combination of $\q^\star(\cdot)$. The second holds true because $\q^\star$ induces steady-state flows that obey the supply circulation property (see~Section \ref{ssec:flow_polytope}); hence, at every node $i$ we have
$$
\sum_{k}\phi_{ki}q_{ki}
=\sum_{k}\sum_{\x\in\mathcal{S}_{n,m}}\pi^\star(\x)\cdot q_{ki}^\star(\x)\phi_{ki}
= \sum_{j}\sum_{\x\in\mathcal{S}_{n,m}}\pi^\star(\x)\cdot q_{ij}^\star(\x)\phi_{ij}=\sum_{j}\phi_{ij}q_{ij},
$$
which shows that $\q$ fulfills the balanced demand property. Hence $\q$ is a feasible solution to the elevated flow relaxation program and the result follows.

Extending to objectives with concave reward curves such as revenue and welfare, we note that Eq.~\eqref{eq:opt_obj_throughput} can be replaced by the following inequality which uses the concavity in the reward curves $R_{ij}(\cdot)$ and applies Jensen's inequality on them.
$$
\opt_m=\sum_{\x\in S_{n,m}}\pi^{\star}(\x)\sum_{i,j}\phi_{ij}R_{ij}\prn*{q^{\star}_{ij}(\x)}\leq \sum_{ij} \phi_{ij} R_{ij}\prn*{\widehat{q}_{ij}}=\widehat{\obj}(\widehat{\q})
$$
The rest of the proof is identical to the case of the throughput objective.
\end{proof}
\vspace{0.5em}
To establish the second step of our analysis, we use the following auxiliary lemma.

\begin{lemma}[balanced demand property implies equal availabilities, \cite{george2011fleet,Waserhole2014}]\vspace{0.1in}
\label{lem:equal_availabilities} 
For any $m$ (including $\infty$) if a state-independent pricing policy $\q$ satisfies the balanced demand property then, at all nodes $i$, the availabilities $A_{i,m}(\q)$ are equal. \end{lemma}
\begin{proof}
Consider $i^\star\in\argmax A_{i,m}(\q)$. Then the balanced demand and supply circulation properties imply 
$$
A_{i^\star,m}(\q)\sum_j \phi_{ji^\star}q_{ji^\star} =A_{i^\star,m}(\q)\sum_j \phi_{i^\star j}q_{i^\star j}=\sum_j A_{j,m}(\q)\phi_{ji^\star}q_{ji^\star}
$$
and thus $\sum_j \big(A_{i^\star,m}(\q)-A_{j,m}(\q)\big)\phi_{ji^\star}q_{ji^\star}=0$. By choice of $i^\star$, each summand is nonnegative, so for each $j$ such that $\phi_{ji^\star}q_{ji^\star}>0$ we obtain $A_{j,m}(\q)=A_{i^\star,m}(\q)$. All availabilities being equal then follows inductively using the assumption that our system is irreducible (see~Appendix \ref{sec:irreducibility}).
\end{proof}
Next we connect the elevated flow relaxation to the infinite-unit system. In fact, we show a stronger statement, that holds for any policy satisfying balanced demand and therefore in particular for the solution of the elevated flow relaxation program. \begin{lemma}[from elevated flow relaxation to infinite-unit state-independent]\vspace{0.1in}
\label{lem:infinite} 
For any state-independent pricing policy $\q$ satisfying the balanced demand property, the value of the elevated objective function of
$\q$ is equal to the value of its objective function in the infinite-unit system
\begin{equation*}
\obj_{\infty}(\q)=\widehat{\obj}(\q).
\end{equation*}
\end{lemma}
\begin{proof}
Since $\q$ satisfies the balanced demand property, by Lemmas \ref{lem:equal_availabilities} and Proposition \ref{prop:infinite}, the availability at all nodes is equal to $1$. This means that (i) the value of the objective function in the infinite-unit limit for pricing policy $\q$ is equal to its elevated value (since no term was increased), and (ii) the flow of customers on each edge is equal to $\phi_{ij}q_{ij}$.
\end{proof}
The surprising step in our results is the pullback step that bounds $\obj_m(\cdot)/\obj_\infty(\cdot)$. To obtain this bound we first show that $\obj_m(\cdot)/\obj_\infty(\cdot)$ is equal to the maximum availability in the $m-$unit system, denoted $A_m(\q) := \max_{i}\prn*{A_{i,m}(\q)}$. Then, we show that $A_m(\q)\geq\frac{m}{m+n-1}$ for all~$\q$.
\begin{lemma}[approximation of finite-unit equals maximum availability]\vspace{0.1in}
\label{lem:max_availability}
The objective function of $\q$ in the $m$-unit system is related to the infinite-limit objective as $$\frac{\obj_{m}(\q)}{\obj_{\infty}(\q)}=r_{\max}(\q)\cdot \frac{G_{m-1}(\q)}{G_m(\q)}=A_{m}(\q).$$
\end{lemma}
\begin{proof}
\proof{Proof.}Let $B_i(\q)=\sum_j \phi_{ij}q_{ij} \cdot I_{ij}(q_{ij})$ denote the contribution of node $i$ to the objective per unit of time in which station $i$ is available.
By substituting $A_{i,m}(\q)=(G_{m-1}(\q)/G_m(\q))\cdot r_i(\q)$, $A_{i,\infty}(\q)=r_i(\q
)/r_{\max}(\q)$, and $B_i(\q)$ into the definition of the objectives in Equation \eqref{eq:objective}, we obtain
\begin{equation*}\label{eq:finite_infinite_obj_ratio}
\frac{\obj_m(\q)}{\obj_\infty(\q)}
= \frac{\sum_i A_{i,m}(\q) B_i(\q)}{\sum_i A_{i,\infty}(\q) B_i(\q)}
= \frac{\frac{G_{m-1}(\q)}{G_m(\q)}\cdot \sum_i r_i(\q) B_i(\q)}{\frac{1}{r_{\max}(\q)}\;\cdot\sum_i r_i(\q) B_i(\q)}
= r_{\max}(\q)\cdot \frac{G_{m-1}(\q)}{G_m(\q)} = A_m(\q),
\end{equation*}
where the last equality follows from the characterization of the availabilities in Equation \eqref{eq:avail}. Note that the argument relies on $OBJ_\infty(\q)\neq 0$ which holds for all policies/settings we consider. In particular, there is always a policy that charges $\epsilon>0$ for every price and achieves a positive objective (since we assumed $F_{ij}(0)<1$).
\end{proof}

\begin{lemma}[weighted bipartite graph among state space of different-unit systems]\label{lem:bipartite}
\vspace{0.1in}
We call $\y\in \mathcal{S}_{n,m-1}$ a neighbor of $\y+e_i\in\mathcal{S}_{n,m}\forall i\in\crl{1,n}$.
There exists a weighted bipartite graph on
$\mathcal{S}_{n,m}\cup\mathcal{S}_{n,m-1}$ such that i) an edge has non-zero weight only if it is connecting neighboring states
, ii) for any vertex corresponding to a state in $\mathcal{S}_{n,m-1}$ the total weight of incident edges is equal to $\frac{m+n-1}{m}$, and iii) for any vertex corresponding to a state in $\mathcal{S}_{n,m}$ the total weight of incident edges is equal to $1$.
\end{lemma}
\begin{proof}
Our construction is shown in Figure \ref{fig:preview}. Each state $\x\in\mathcal{S}_{n,m}$ is adjacent to $\y=\x-e_i\in\mathcal{S}_{n,m-1}$ for all $i$ with $x_i>0$. On these edges, the weight is $\omega_{\x\y}=\frac{x_i}{m}$. Thus, the total weight incident to $\x$ is $\sum_{\y}\omega_{\x\y}=\sum_i\frac{x_i}{m}=1$. On the other hand, each state $\y\in\mathcal{S}_{n,m-1}$ is adjacent to the states $\x=\y+e_i\;\forall i\in[n]$. The respective weight incident on $\y$ is $\sum_x\omega_{\x\y}=\sum_i \frac{y_i+1}{m}=\frac{m-1+n}{m}$. Finally, there is only weight on edges between neighboring states. This concludes the proof of the lemma.
\end{proof}

\begin{figure}[!h]
\centering
\subfigure[Graph between $\mathcal{S}_{2,3}$ and $\mathcal{S}_{2,2}$]{
\label{fig:bigraph_a}
\centering
\scalebox{1.2}{
\tikzset{every picture/.style={line width=0.75pt}} 

\begin{tikzpicture}[x=0.75pt,y=0.75pt,yscale=-1,xscale=1]

\draw   (469,83) .. controls (469,74.72) and (475.72,68) .. (484,68) .. controls (492.28,68) and (499,74.72) .. (499,83) .. controls (499,91.28) and (492.28,98) .. (484,98) .. controls (475.72,98) and (469,91.28) .. (469,83) -- cycle ;
\draw   (468,43) .. controls (468,34.72) and (474.72,28) .. (483,28) .. controls (491.28,28) and (498,34.72) .. (498,43) .. controls (498,51.28) and (491.28,58) .. (483,58) .. controls (474.72,58) and (468,51.28) .. (468,43) -- cycle ;
\draw   (470,124) .. controls (470,115.72) and (476.72,109) .. (485,109) .. controls (493.28,109) and (500,115.72) .. (500,124) .. controls (500,132.28) and (493.28,139) .. (485,139) .. controls (476.72,139) and (470,132.28) .. (470,124) -- cycle ;
\draw   (471,166) .. controls (471,157.72) and (477.72,151) .. (486,151) .. controls (494.28,151) and (501,157.72) .. (501,166) .. controls (501,174.28) and (494.28,181) .. (486,181) .. controls (477.72,181) and (471,174.28) .. (471,166) -- cycle ;
\draw   (570,60) .. controls (570,51.72) and (576.72,45) .. (585,45) .. controls (593.28,45) and (600,51.72) .. (600,60) .. controls (600,68.28) and (593.28,75) .. (585,75) .. controls (576.72,75) and (570,68.28) .. (570,60) -- cycle ;
\draw   (571,102) .. controls (571,93.72) and (577.72,87) .. (586,87) .. controls (594.28,87) and (601,93.72) .. (601,102) .. controls (601,110.28) and (594.28,117) .. (586,117) .. controls (577.72,117) and (571,110.28) .. (571,102) -- cycle ;
\draw   (572,141) .. controls (572,132.72) and (578.72,126) .. (587,126) .. controls (595.28,126) and (602,132.72) .. (602,141) .. controls (602,149.28) and (595.28,156) .. (587,156) .. controls (578.72,156) and (572,149.28) .. (572,141) -- cycle ;
\draw    (498,43) -- (570,60) ;
\draw    (499,83) -- (570,60) ;
\draw    (499,83) -- (571,102) ;
\draw    (500,124) -- (571,102) ;
\draw    (500,124) -- (572,141) ;
\draw    (501,166) -- (572,141) ;

\draw (472,35) node [anchor=north west][inner sep=0.75pt]   [align=left] {3,0};
\draw (473,75) node [anchor=north west][inner sep=0.75pt]   [align=left] {2,1};
\draw (473,116) node [anchor=north west][inner sep=0.75pt]   [align=left] {1,2};
\draw (475,158) node [anchor=north west][inner sep=0.75pt]   [align=left] {0,3};
\draw (574,52) node [anchor=north west][inner sep=0.75pt]   [align=left] {2,0};
\draw (574,94) node [anchor=north west][inner sep=0.75pt]   [align=left] {1,1};
\draw (576,133) node [anchor=north west][inner sep=0.75pt]   [align=left] {0,2};
\draw (528,34.4) node [anchor=north west][inner sep=0.75pt]  [font=\small]  {$1$};
\draw (516,57.4) node [anchor=north west][inner sep=0.75pt]  [font=\footnotesize]  {$1/3$};
\draw (535,80.4) node [anchor=north west][inner sep=0.75pt]  [font=\footnotesize]  {$2/3$};
\draw (515,99.4) node [anchor=north west][inner sep=0.75pt]  [font=\footnotesize]  {$2/3$};
\draw (535,120.4) node [anchor=north west][inner sep=0.75pt]  [font=\footnotesize]  {$1/3$};
\draw (525,141.4) node [anchor=north west][inner sep=0.75pt]  [font=\footnotesize]  {$1$};
\end{tikzpicture}
}
}
\hspace{1cm}
\subfigure[Construction for general $n,m$]{
\label{fig:bigraph_b}
\centering
\scalebox{1}{
\tikzset{every picture/.style={line width=0.75pt}} 

\begin{tikzpicture}[x=0.75pt,y=0.75pt,yscale=-1,xscale=1]

\draw   (118.67,161.33) .. controls (118.67,156.92) and (122.25,153.33) .. (126.67,153.33) -- (158,153.33) .. controls (162.42,153.33) and (166,156.92) .. (166,161.33) -- (166,185.33) .. controls (166,189.75) and (162.42,193.33) .. (158,193.33) -- (126.67,193.33) .. controls (122.25,193.33) and (118.67,189.75) .. (118.67,185.33) -- cycle ;
\draw   (239.67,97.33) .. controls (239.67,92.92) and (243.25,89.33) .. (247.67,89.33) -- (302.67,89.33) .. controls (307.08,89.33) and (310.67,92.92) .. (310.67,97.33) -- (310.67,121.33) .. controls (310.67,125.75) and (307.08,129.33) .. (302.67,129.33) -- (247.67,129.33) .. controls (243.25,129.33) and (239.67,125.75) .. (239.67,121.33) -- cycle ;
\draw   (239.67,150.33) .. controls (239.67,145.92) and (243.25,142.33) .. (247.67,142.33) -- (302.67,142.33) .. controls (307.08,142.33) and (310.67,145.92) .. (310.67,150.33) -- (310.67,174.33) .. controls (310.67,178.75) and (307.08,182.33) .. (302.67,182.33) -- (247.67,182.33) .. controls (243.25,182.33) and (239.67,178.75) .. (239.67,174.33) -- cycle ;
\draw   (240.67,225.33) .. controls (240.67,220.92) and (244.25,217.33) .. (248.67,217.33) -- (303.67,217.33) .. controls (308.08,217.33) and (311.67,220.92) .. (311.67,225.33) -- (311.67,249.33) .. controls (311.67,253.75) and (308.08,257.33) .. (303.67,257.33) -- (248.67,257.33) .. controls (244.25,257.33) and (240.67,253.75) .. (240.67,249.33) -- cycle ;
\draw    (158,153.33) -- (239.67,108.67) ;
\draw    (165.67,166.67) -- (239.67,163.67) ;
\draw    (158,193.33) -- (239.67,240.67) ;

\draw (134,165.73) node [anchor=north west][inner sep=0.75pt]  [font=\Large]  {$\mathbf{{\displaystyle x}}$};
\draw (245.67,102.73) node [anchor=north west][inner sep=0.75pt]  [font=\Large]  {${\displaystyle \mathbf{x} -e_{1}}$};
\draw (245.67,153.73) node [anchor=north west][inner sep=0.75pt]  [font=\Large]  {${\displaystyle \mathbf{x} -e_{2}}$};
\draw (266,186.07) node [anchor=north west][inner sep=0.75pt]  [font=\Large]  {$\vdots $};
\draw (245.67,230.73) node [anchor=north west][inner sep=0.75pt]  [font=\Large]  {${\displaystyle \mathbf{x} -e_{n}}$};
\end{tikzpicture}
}
}
\caption[Bipartite graph construction]{Bipartite graph construction between states in $\mathcal{S}_{n,m}$ and states in $\mathcal{S}_{n,m-1}$, as described in Lemma \ref{lem:bipartite}. Fig. \ref{fig:bigraph_a} shows the construction of the bipartite graph between $\mathcal{S}_{2,3}$ and $\mathcal{S}_{2,2}$ (i.e., between configurations of 3 units in 2 nodes, and 2 units in 2 nodes). Fig. \ref{fig:bigraph_b} shows the general construction. Note that the sum of weights of incident edges for any node on the left (i.e. any state in $\mathcal{S}_{n,m}$) is $1$, while it is $(m+n-1)/m$ for nodes on the right (i.e. states in $\mathcal{S}_{n,m-1}$). 
}
\label{fig:preview}
\end{figure}

\begin{lemma}[from infinite-unit to finite-unit state-independent]
\label{lem:finite}\vspace{0.1in}
For any state-independent pricing policy $\q$, the maximum availability in the $m$-unit system is at least $\nicefrac{m}{(m+n-1)}$, i.e.,
\begin{equation*}
A_m(\q)=r_{\max}(\q)\cdot \frac{G_{m-1}(\q)}{G_m(\q)}\geq\frac{m}{m+n-1}.
\end{equation*}
\end{lemma}
\begin{proof}
In order to uniformly bound the above expression, we apply
the weighted bipartite graph, between the states in $\mathcal{S}_{n,m-1}$ and the states in $\mathcal{S}_{n,m}$, described in Lemma~\ref{lem:bipartite}. Following the same notation as before, we denote the weight between states $\x\in\mathcal{S}_{n,m}$ and $\y\in\mathcal{S}_{n,m-1}$ by $\omega_{\x\y}$. Recall that non-zero weights only exist between neighboring states, i.e. when $\x=\y+e_i\in\mathcal{S}_{n,m}$ for some $i$; further, the total weight of edges incident to any state in $\x\in\mathcal{S}_{n,m}$ is equal to $\sum_{\y}\omega_{\x\y}=1$, and the total weight of edges incident to any state in $\y\in\mathcal{S}_{n,m-1}$ is equal to $\sum_{\x}\omega_{\x\y}=\frac{m+n-1}{m}$.
\begin{align*}
 r_{\max}(\p)\cdot\frac{G_{m-1}(\q)}{G_m(\q)}&=r_{\max}(\q)\frac{
\sum_{\y\in \mathcal{S}_{n,m-1}} \prod_{j=1}^n\prn*{r_j(\q)}^{y_j}}{\sum_{\x\in \mathcal{S}_{n,m}} \prod_{j=1}^n\prn*{r_j(\q)}^{x_j}}\\
&=r_{\max}\q)\cdot\frac{
\sum_{\y\in \mathcal{S}_{n,m-1}} \prod_{j=1}^n\prn*{r_j(\q)}^{y_j}}{\sum_{\x\in \mathcal{S}_{n,m}} \left(\sum_{\y\in \mathcal{S}_{n,m-1}} \omega_{\x\y}\right)\prod_{j=1}^n\prn*{r_j(\q)}^{x_j}}\\
&=r_{\max}(\q)\cdot\frac{
\sum_{\y\in \mathcal{S}_{n,m-1}} \prod_{j=1}^n\prn*{r_j(\q)}^{y_j}}
{\sum_{(\x,\y})\in \mathcal{S}_{n,m}\times}
\mathcal{S}_{n,m-1}  \omega_{\x\y}\prod_{j=1}^n\prn*{r_j(\q})^{y_j+(x_j-y_j)}\\
&\geq 
\frac{
\sum_{\y\in \mathcal{S}_{n,m-1}} \prod_{j=1}^n\prn*{r_j(\q)}^{y_j}}{\sum_{\y\in \mathcal{S}_{n,m-1}} \left(\sum_{\x\in \mathcal{S}_{n,m}} \omega_{\x\y}\right)\prod_{j=1}^n\prn*{r_j(\q)}^{y_j}}\\
&=\frac{\sum_{\y\in \mathcal{S}_{n,m-1}}\prod_{j=1}^n\prn*{r_j(\q)}^{y_j}}
{\left(\frac{m+n-1}{m}\right)
\sum_{\y\in \mathcal{S}_{n,m-1}}\prod_{j=1}^n\prn*{r_j(\q)}^{y_j}}
= \frac{m}{m+n-1}
\end{align*}
The third equality holds as $\sum_{\y}\omega_{\x\y}=1$; the second-to-last follows from $\sum_{\x} \omega_{\x\y}=\frac{m+n-1}{m}$. Crucially, $\omega_{\x\y}>0$ only holds for neighboring states $\x\in\mathcal{S}_{n,m}$ and $\y\in\mathcal{S}_{n,m-1}$ implying the inequality.
\end{proof}
\begin{theorem}[approximation guarantee for pricing with general objectives]\label{thm:pricing} Consider any objective function $\obj_m$ for the $m$-unit system with concave reward curves $R_{ij}(\cdot)$.
Let $\widetilde{\q}$ be the pricing policy returned by Algorithm \ref{alg:demand_circulation_reduction_general_pricing} and $\opt_m$ be the value of the objective function for the optimal state-dependent pricing policy in the $m$-unit system. Then
\begin{equation}
\obj_m(\widetilde{\q})\geq \frac{m}{m+n-1}\opt_m.
\end{equation}
\end{theorem}
\begin{proof}
The proof follows by direct applications of Lemmas \ref{lem:state_dependent}, \ref{lem:infinite}, \ref{lem:max_availability}, and \ref{lem:finite}.
\end{proof}
\vspace{0.5em}

We remark that Lemma \ref{lem:finite} is stronger than what is needed in the proof of Theorem \ref{thm:pricing} as it applies to any state-independent $\q$, not just pricing policies that fulfill the balanced demand property (as is guaranteed in Algorithm \ref{alg:demand_circulation_reduction_general_pricing}). Indeed, when $\q$ fulfills the property, it has been known since the 1980s \cite{whitt1984} that 
$A_{i,m}(\q)=\frac{m}{m+n-1}\forall i$, and this has been used by  Waserhole and Jost who originally, for throughput, showed the approximation guarantee we reproved in this section \cite{Waserhole2014}. 
It follows that, when pricing is unrestricted as a control, we can always set quantiles that fulfill the balanced demand property and our guarantees can be derived as
\begin{equation*}
\obj_m(\q)= \sum_{i}A_{i,m}(\q)\cdot\prn*{\sum_j \phi_{ij}q_{ij}\cdot I_{ij}(q_{ij})} = \frac{m}{m+n-1}\prn*{\sum_{i,j} \phi_{ij}q_{ij}\cdot I_{ij}(q_{ij})}=\frac{m}{m+n-1}\widehat{\obj}(\q)
\end{equation*}
without including a detour to the infinite-unit limit (Lemma \ref{lem:infinite} to \ref{lem:finite}). This is because when~$A_{i,m}(\q)=\frac{m}{m+n-1}$ for all $i$, availability can be factored out. In the next section we extend the bounds for Gordon-Newell networks to BCMP networks to capture travel-delays. Then, unavailability is not only due to units being at the wrong locations, but also due to units being in transit. This requires new bounds for~$A_m(\q)$, though as long as demand is balanced, with $A_{i,m}(\q)$ equal for all $i$, one can still factor out availability terms to avoid the infinite-unit limit.
However, in Section \ref{sec:applications} we consider scenarios in which the balanced demand property cannot be guaranteed because restrictions on pricing make it suboptimal or even infeasible. In these settings the availability term cannot be factored out, yet the infinite-unit limit continues to provide guarantees based on bounding~$A_m(\q)$.

\section{Incorporating Travel Times Between Nodes}
\label{ssec:delay}

In the previous section with no travel-times, we obtained an approximation guarantee of $m/(m+n-1)$ with respect to the objective of the elevated flow relaxation, independent of demand arrival rates. In contrast, the policies we give in this section have approximation guarantees that degrade with increasing demand. This is similar in spirit to heavy-traffic behavior in open queueing networks, as in the presence of travel-times, there is additional unavailability due to units being in transit between nodes. Nevertheless, we design policies which guarantee convergence to the elevated-flow relaxation upper bound at a rate $O(\frac{1}{\sqrt{m}})$. We also prove that this is tight. Moreover, we characterize how the approximation guarantees degrade from $O(1/m)$ to $O(1/\sqrt{m})$ as the system goes from over-supplied to over-loaded regimes.

We begin in Section \ref{ssec:availability_w_delays} by bounding $A_m(\q)$ in the presence of travel-times, similar to  Lemma \ref{lem:finite}. Thereafter, in Section \ref{ssec:pricing_delays} we propose a pricing policy whose performance converges to the objective of the elevated flow relaxation at a $O(1/m)$ rate in lightly-loaded settings, and $O(1/\sqrt{m})$ rate under heavy loading (that is, when both the demand and the supply are scaled). We complement this by demonstrating that faster asymptotic rates of convergence (to the relaxation objective) in $m$ are not achievable even by state-dependent policies.  

\subsection{Bounding Availability in the Presence of Travel Times}\label{ssec:availability_w_delays}
In the setting with delays the resulting flows $f_{ij,m}(\q)$ continue to satisfy demand bounding and supply circulation for the same reasons as before, that is, if supply circulation did not hold there would be nodes at which more than $m$ units would accumulate. Consequently, the elevated flow relaxation (see Algorithm \ref{alg:demand_circulation_reduction_general_pricing}) continues to provide an upper bound. 
Moreover, adding link queues does not affect the optimization problems we consider in the infinite-unit system; in particular, Lemma~\ref{lem:infinite} also continues to hold in this setting. Finally, from Lemma~\ref{lem:max_availability} (since Lemma \ref{lem:avail} holds for BCMP networks, the proof of the lemma extends with $Z_m$ and $Z_{m-1}$ replacing $G_m$ and $G_{m-1}$), we know that the ratio of objectives between the infinite-unit system and the finite-unit system equals the maximum availability, among all nodes, in the finite-unit system, i.e.
$
\frac{\obj_m(\q)}{\obj_\infty(\q)}=A_{m}(\q)
$. In order to obtain an approximation ratio, we now need to understand how $A_{m}(\q)$ changes when link queues are added. Let $M$ denote the random variable corresponding to the steady-state number of \emph{available} (i.e. not in transit) units across all nodes; we first show that for a given $M$, the distribution of units across nodes is the same as in a Gordon-Newell network with $M$ units.
\begin{lemma}\vspace{0.1in}.
\label{lem:steadystateM}
Conditioned on $M$, the distribution of $\{X_i\}_{i\in[n]}$ in the network with travel-times is identical to an $n$-node $M$-unit Gordon-Newell network with the same quantiles and arrival rates. 
\end{lemma}
\begin{proof}
This follows directly from the product-form nature of the steady-state distribution in Equation~\eqref{eq:steadystateBCMP}, since the denominator in 
$$\pi_{\y,m}(\q|M)
= \frac{\pi_{\y,m}(\q)}{\sum_{\x:\sum_{i=1}^nx_i=M}\pi_{\x,m}(\q)}
= \frac{\prod_{i\in[n]} \prn*{r_i(\q)}^{y_i}}{
\sum_{\x:\sum_{i=1}^nx_i=M}
\prod_{i\in[n]} \prn*{r_i(\q)}^{x_i}
}
$$
is exactly equal to the normalizing constant in the Gordon-Newell Theorem (see~Equation \ref{eq:steadystate}).
\end{proof}

Using this, we can now obtain a bound for the $m$-unit system availability in the next lemma. Note that the bound converges to 1 as $m\to\infty$. Furthermore, it is crucial to observe that for $m\geq100$, the only assumption of the Lemma is that  $\sum_{ij}\phi_{ij}\tau_{ij}q_{ij} \leq  m-2\sqrt{m\ln(m)}$, i.e., when considering an asymptotic regime in which both the supply $m$ and the demand $\{\phi_{ij}\}$ grow large, our results hold when both grow at the same rate.

\begin{lemma}
\label{lem:Am_w_delay}
\vspace{0.1in} For any network with parameters $\crl*{\phi_{ij},F_{ij}(\cdot),\tau_{ij}}$
if $m\geq 100$ and quantiles $\q$ satisfy $\sum_{ij}\phi_{ij}\tau_{ij}q_{ij} \leq  m-2\sqrt{m\ln(m)}$ then
\begin{align*}
A_m(\q)\geq
\prn*{1 - \frac{3}{\sqrt{m}}}\prn*{\frac{\sqrt{m\ln m}}{\sqrt{m\ln m}+n-1}}.
\end{align*}
\end{lemma}

\begin{proof}
First, for any given policy $\q$, as before we have the realized flows $f_{ij,m}(\q) = q_{ij}\phi_{ij}A_{i,m}(\q)$; moreover, this is the expected rate of units entering link queue $X_{ij}$. Let $D = m-M$ be the number of units which are in transit. Now, by Little's law (see~\cite{kelly2011} or \cite{IntroStochNetworks}), we have that the expected number of units in link queues is given by $\mathbb{E}\left[D\right]=\sum_{i,j}A_{i,m}(\q)\phi_{ij}q_{ij}\tau_{ij}$. 

Note that the link queues $\{X_{ij}\}$ are  (first-order) stochastically dominated by independent $M/M/\infty$ queues; formally, for any link queue $X_{ij}$, if $Y_{ij}$ is the stationary distribution of an $M/M/\infty$ queue with arrival rate $\phi_{ij}q_{ij}$ and average transition time $\tau_{ij}$, then we have that $P\left[Y_{ij}\geq x\right]\geq P\left[X_{ij}\geq x\right]\;\forall x\in\mathbb{Z}$. 
This follows from a simple coupling argument, where incoming customers follow an independent Poisson process of rate $\phi_{ij}q_{ij}$ and enter the link queue with a \emph{virtual unit}, irrespective of whether the customer engages a unit or not in the real system. Thus $D$ is stochastically dominated by $\widetilde{D}=Poi(\sum_{i,j}\phi_{ij}q_{ij}\tau_{ij})$. Further, since $D$ is bounded above by $m$, $D$ is also stochastically dominated by $\widehat{D} = \min\{\widetilde{D},m\}$.

Next, from Lemma \ref{lem:steadystateM}, we know that conditioned on there being $M$ available units in the steady-state system, the distribution of units in node queues is identical to that of an $n$-node $M$-unit Gordon-Newell network. Denoting by $A_{m}(\q|M)\triangleq\max_{i\in[n]}\PP\brk*{X_i>0|M}$ the maximum availability across all nodes conditioned on there being $M$ available units, we know  from Lemma \ref{lem:finite} that for any $n$-node, $M$-unit Gordon-Newell network the maximum availability is at least $M/(M+n-1)$. Since~$M=m-D$ and $\nicefrac{(m-x)}{(m+n-1-x)}$ is decreasing in $x$ for $x\leq m$,
it follows that
\begin{align}\label{eq:delays}
A_{m}(\q) &\geq \mathbb{E}\left[\frac{m-D}{m+n-1-D}\right]
\geq \mathbb{E}\left[\frac{m-\widehat{D}}{m+n-1-\widehat{D}}\right].
\end{align}
Further, by definition of $\widehat{D}$ we observe that
$\PP\brk*{\widehat{D}> m\prn*{1-\sqrt{\frac{\ln m}{m}}}}=\PP\brk*{\widetilde{D}> m\prn*{1-\sqrt{\frac{\ln m}{m}}}}$. We can now apply a standard Chernoff bound for the Poisson random variable $\widetilde{D}$ (see~from Lemma~\ref{lem:poissontail} in Appendix~\ref{app:proofs_point}), using the assumption that $ m - 2\sqrt{m\ln(m)}\geq \sum_{ij}\phi_{ij}\tau_{ij}q_{ij} = \EE[\widetilde{D}]$. In particular,  we may bound $\PP\brk*{\widehat{D}> m\prn*{1-\sqrt{\frac{\ln m}{m}}}}$ by
\begin{align}
\PP\brk*{\widetilde{D}>m\prn*{1-\sqrt{\frac{\ln m}{m}}}}
&\leq \exp\prn*{\frac{-m\ln m}{2(m-2\sqrt{m\ln m})}\prn*{1-\frac{\sqrt{m\ln m}}{m-2\sqrt{m\ln m}}}}\nonumber \\
&= \exp\prn*{\frac{-m\ln m\prn*{m-3\sqrt{m\ln m}}}{2(m-2\sqrt{m\ln m})^2}}\nonumber
\\
&= \exp\prn*{ \frac{-\ln m\cdot \prn*{1-3\sqrt{\ln m/m}}}{2\prn*{1-4\sqrt{\ln m/m}\prn*{1-\sqrt{\ln m/m}}}}}\nonumber\\
&\leq \exp\prn*{ \frac{-\ln m \prn*{1-3\sqrt{\ln m/m}}}{2\prn*{1-1.5\sqrt{\ln m/m}}}}\nonumber \\
&= \exp\prn*{ \frac{-\ln m}{2} \prn*{1 - \frac{3\sqrt{\ln m/m}}{\prn*{2-3\sqrt{\ln m/m}}}}}
\nonumber\\
&= \frac{1}{\sqrt{m}}\exp\prn*{ \frac{3\ln m}{4\sqrt{m/\ln m}-6}}\leq\frac{3}{\sqrt{m}}\quad\mbox{for $m\geq 100$}.\nonumber
\label{eq:Dbound}
\end{align}
Here, the second inequality holds true because $4\sqrt{\ln m/m}(1-\sqrt{\ln m/m}) \geq 1.5\sqrt{\frac{\ln m}{m}}$ and the final inequality holds true because $\exp\prn*{ \frac{3\ln m}{4\sqrt{m/\ln m}-6}}\leq 3
$ for $m\geq 100$.
We can use the above to bound the availability in Inequality \eqref{eq:delays} as
\begin{align*}
A_{m}(\q) 
&\geq\prn*{1 - \frac{3}{\sqrt{m}}}\prn*{\frac{m-(m-\sqrt{m\ln m})}{m-(m-\sqrt{m\ln m})+n-1}} + \frac{3}{\sqrt{m}}\cdot 0.
\end{align*}
Simplifying, we obtain the result.
\end{proof}

\subsection{Pricing in the Presence of Travel Times}\label{ssec:pricing_delays}
We are now ready to extend our pricing policy to the setting with transit delays. In order to do so, we need to first extend the elevated flow relaxation by adding an extra constraint. The main observation is that in an $m$-unit system with transit delays, there is an additional \emph{conservation constraint} induced by the fact that the number of units in the link queue cannot exceed $m$. As before, let $f_{ij}^{m}(\q) = \widehat{q}_{ij}\phi_{ij}$ denote the expected rate of units entering link queue $X_{ij}$; then by Little's law (see~\cite{kelly2011} or \cite{IntroStochNetworks}), we have that the expected number of units in link queues is given by $\sum_{i,j}\phi_{ij}\widehat{q}_{ij}\tau_{ij}$, which, in an $m$-unit system, must be bounded by $m$.
To incorporate this, we need to add an additional \emph{rate-limiting constraint} to the elevated flow relaxation wherein we ensure that $\sum_{i,j}\phi_{ij}\widehat{q}_{ij}\tau_{ij}\leq m$.  
This gives us the \emph{rate-limited elevated flow relaxation program} in Algorithm~\ref{alg:delay_efr}.
\begin{algorithm}
\caption{The Rate-Limited Elevated Flow Relaxation Program}
\label{alg:delay_efr}
\begin{algorithmic}[1]
\REQUIRE arrival rates $\phi_{ij}$, value distributions $F_{ij}$,  reward curves~$R_{ij}$, scaling parameter $\varepsilon_m$, travel-times $\tau_{ij}$.
\STATE Find $\crl*{{q}'_{ij}}$ that solve the following relaxation:
$$
\begin{array}{rcll}
\max_\q &&  \sum_{(i,j)} \phi_{ij}R_{ij}({q}_{ij}) &\\
\sum_{k} \phi_{ki}{q}_{ki}  & = & \sum_{j} \phi_{ij}{q}_{ij} & \forall\;i\\
\sum_{i,j}\phi_{ij}\tau_{ij}{q}_{ij} &\leq& m& \\
{q}_{ij} & \in & [0,1] & \forall\;i,j.
\end{array}
$$
\STATE If $\sum_{i,j}\phi_{ij}\tau_{ij}{q}_{ij}' \geq m(1-\epsilon_m)$, set $\widetilde{q}_{ij}= q'_{ij}\cdot\prn*{1 - \varepsilon_m}$, else $\widetilde{q}_{ij}= q'_{ij}$
\STATE Output \emph{state-independent} prices $\widetilde{p}_{ij} = F_{ij}^{-1}(1-\widetilde{q}_{ij})$ and the corresponding quantiles $\widetilde{q}_{ij}$
\end{algorithmic}
\end{algorithm}

\begin{theorem}[approximation guarantee in the presence of delays]\vspace{0.1in}
\label{thm:pricing_w_delays}
Consider any objective function $\obj_m$ for the $m$-unit system with concave reward curves $R_{ij}(\cdot)$. Let quantiles $\widetilde{\q}$ be the output of Algorithm \ref{alg:delay_efr} with input $\varepsilon_m\defeq2\sqrt{\ln m/m}$, $\opt_m$ be the value of the objective function for the optimal state-dependent pricing policy, and  $m\geq 100$. Then
\begin{equation*}
\frac{\obj_m(\widetilde{\q})}{\opt_m}\geq \prn*{1-\varepsilon_m}\prn*{\frac{\sqrt{m\ln m}}{\sqrt{m\ln m}+n-1} - \frac{3}{\sqrt{m\ln m}}}. 
\end{equation*}
\end{theorem}

\begin{proof}
The proof follows a similar roadmap as that of Theorem \ref{thm:pricing}. In particular, we argue that
\begin{enumerate}
\item the rate-limited elevated flow relaxation provides an upper bound for any state-dependent policy,
\item the rate-limited elevated flow relaxation solution is achieved by a state-independent policy in the infinite-unit system, and
\item the ratio of the performance of any state-independent policy $\q$ in the infinite-unit and $m$-unit system is equal to the maximum availability $A_m(\q)$.
\end{enumerate}

First, similar to Lemma \ref{lem:state_dependent}, note that since the realized flows in the $m$-unit system must obey the conservation laws encoded by the rate-limited elevated flow relaxation, hence $\opt_m$ is bounded by the solution of the rate-limited elevated flow relaxation $\sum_{(i,j)} \phi_{ij}R_{ij}(q_{ij})$.
Moreover, since per-ride rewards $I_{ij}(\cdot)$ are non-increasing in $\q$,  scaling the $q_{ij}$ by $(1-\varepsilon_m)$ results in an elevated objective value that obeys
$$(1-\varepsilon_m)\sum_{(i,j)} \phi_{ij}R_{ij}(q_{ij}) \leq \sum_{(i,j)} \phi_{ij}R_{ij}(\widetilde{q}_{ij}).$$
Further, $\sum_{i,j}\phi_{ij}\widetilde{q}_{ij}\tau_{ij}\leq m\cdot\prn*{1 - \varepsilon_m}$.
Now, using similar arguments as in Lemma \ref{lem:infinite}, we can show that using a state-independent policy $\widetilde{\q}$ in the infinite-unit limit gives $\obj_{\infty}(\widetilde{\q}) = \sum_{(i,j)} \phi_{ij}R_{ij}(\widetilde{q}_{ij})$ (note that we use the same $\widetilde{\q}$ as derived from the $m$ unit rate-limited elevated flow relaxation in the infinite unit limit; in other words, we scale the number of units to infinity, but retain the constraint $\sum_{i,j}\phi_{ij}\tau_{ij}\widehat{q}_{ij} \leq m$ for a fixed $m$). Next, from Lemma \ref{lem:max_availability}, we get that $\obj_m(\widetilde{\q}) = A_m(\widetilde{\q})\obj_{\infty}(\widetilde{\q})$. Finally, using Lemma \ref{lem:Am_w_delay}, we get the desired bound
$$\frac{\obj_m(\widetilde{\q})}{\opt_m}\geq \prn*{1-\varepsilon_m}\prn*{\frac{\sqrt{m\ln m}}{\sqrt{m\ln m}+n-1} -\frac{3}{\sqrt{m}}}.\hfill
$$
\end{proof}

Note that for any fixed $n$, the theorem shows that the policy returned by the rate-limited elevated flow relaxation is asymptotically optimal as $m\rightarrow\infty$ for \emph{any demand rates and transit delays} $\crl*{\phi_{ij},\tau_{ij}}$, including in particular demand values $\phi_{ij}^m$ that scale together with the supply. In Appendix \ref{app:no_prices} we use this to give a finite-$m$ characterization for the asymptotic results in~\cite{braverman2016empty}.

It is now natural to wonder whether the guarantee in Theorem \ref{thm:pricing_w_delays} is tight. In particular, the rate of convergence is only in the square-root of $m$ and as such it converges an order of magnitude slower than our guarantees without travel-times. It turns out that there are two answers to that question:
\begin{itemize}
    \item when there is slack in the rate-limiting constraint, i.e., $\sum_{i,j}\phi_{ij}\tau_{ij}q_{ij}<m(1-c)$ for some $c>0$, the rate of convergence is linear in m (see Corollary \ref{cor_lin_rate}). This occurs, e.g., when there are no travel delays as in Theorem \ref{thm:pricing}, when demand does not grow with $m$, or when revenue-maximizing prices are high enough;
    \item when the rate-limiting constraint is tight, and demand grows with $m$, then the rate of convergence of our algorithm is $O(\frac{1}{\sqrt{m}})$ (Theorem \ref{thm:pricing_w_delays}). In general, $o(1/\sqrt{m})$ convergence to the elevated flow relaxation, cannot be attained even by an optimal policy (see Proposition \ref{prop_sqrt_lower_bound}).
\end{itemize}
Combined with results in the next section, these results shed more light on numerical observations of Braverman et al. \cite{braverman2016empty}, who found a convergence to optimality at rate $(1/\sqrt{m})$. We prove that this holds true, and is optimal in general. However, if at least one node has a constant fraction of units in their optimization program's optimal solution (in the language from their Lemma 2: if at least one $i$, satisfies  $\bar{e}^\star_{ii}>0$), Corollary \ref{cor_lin_rate} implies a rate of $O(1/m)$, i.e.,  faster than empirically observed.

\begin{corollary}\label{cor_lin_rate}\vspace{0.1in}
Suppose that the optimal solution $\q^\star$ to the flow relaxation in Algorithm~\ref{alg:delay_efr} is not constrained by the rate-limited constraint, that is, $\sum_{ij}\phi_{ij}\tau_{ij}q^\star_{ij}=(1-2c)m$ for some $c>0$. Then, the rate of convergence to the upper bound is linear. Specifically, for $m\geq 100$, 
\begin{equation*}
\frac{\obj_m(\q^\star)}{\opt_m}
\geq 
\prn*{1-e^{-\frac{c^2}{2}m}}
\prn*{\frac{cm}{cm+n-1}}. 
\end{equation*}
\end{corollary}

\begin{proof}
By a Chernoff bound (Lemma \ref{lem:poissontail}), with $\widehat{D}$ and $\widetilde{D}$ defined as in Lemma \ref{lem:Am_w_delay}, we have
 \begin{equation*}
 \PP\brk*{\widehat{D}> m\prn*{1-c}m}
 =\PP\brk*{\widetilde{D}> m\prn*{1-c}m}
 \leq e^{\prn*{-\frac{(cm)^2}{2(1-2c)m}\prn*{1-\frac{cm}{(1-2c)m}}}}
 = e^{\prn*{-\frac{c^2m}{2(1-c)}\prn*{1-\frac{c}{1-2c}}}}
 \end{equation*}
Simplifying, we find that  $\prn*{\frac{c^2}{2(1-c)}\prn*{1-\frac{c}{1-2c}}}\geq \frac{c^2}{2}$. Thus, with probability at least $\prn*{1-e^{-\frac{c^2}{2}m}}$, there are at least $cm$ units not in transit. The rest of the proof is equivalent to that in Lemma \ref{lem:Am_w_delay} and Theorem \ref{thm:pricing_w_delays}.
\end{proof}

\begin{prop}\label{prop_sqrt_lower_bound}\quad 
There exist regimes in which the gap between the optimal state-dependent policy and the objective of the rate-limited elevated flow relaxation is bounded by $\Omega\left(\frac{1}{\sqrt{m}}\right)$.
\end{prop}
\begin{proof}Consider a system consisting of just a single node $i$, in which customers travel from $i$ to itself and $\tau_{ii}=1$. Suppose our objective is to maximize throughput; in this case, it is easy to see that  setting $q_{ii}=1$ is the optimal policy. Further, suppose that $\phi_{ii}=m$, i.e., the rate of demand is exactly equal to, and scales with, the supply available.  Then, using Equation \eqref{eq:steadystateBCMP} we can write the objective  as follows:
\begin{equation*}
 \obj_m = m A_{i,m}(\q) =
 m\brk*{1-
 \prn*{\frac{
 \prn*{\frac{\frac{1}{2}}{m}}^{0}
\prn*{\frac{\prn*{\frac{1}{2}}^{m}}{m!}}
 }{
 \sum_{x_i=0}^{m}
\prn*{\frac{\frac{1}{2}}{m}}^{x_i}
\prn*{\frac{\prn*{\frac{1}{2}}^{m-x_i}}{(m-x_i)!}}}}}
=
 m\brk*{1-
 \prn*{
 \frac{
\prn*{\frac{1}{m!}}
 }
 {
 \sum_{x_i=0}^{m}
\prn*{\frac{1}{m}}^{x_i}
\prn*{\frac{1}{(m-x_i)!}}}}.
}
\end{equation*}

In order to prove our claim that no policy can obtain a rate of convergence that is faster than in the square-root of $m$, we need to bound the subtrahend by $\Omega(\frac{1}{\sqrt{m}})$:
\begin{equation*}
  \frac{ \prn*{\frac{1}{m!}}}{\sum_{x_i=0}^{m}
\prn*{\frac{1}{m}}^{x_i-m}
\prn*{\frac{1}{(m-x_i)!}}}\frac{1}{m^m}
=\frac{m^m}{m!}\frac{1}{\sum_{i=0}^m \frac{m^{x_i}}{x_{i}!}}
\geq \frac{m^m}{m!}\frac{1}{\sum_{i=0}^\infty \frac{m^{x_i}}{x_{i}!}}
= \frac{e^{-m}m^m}{m!}
= \Omega\prn*{\frac{1}{\sqrt{m}}}.
\end{equation*}
Here, the final bound follows from Sterling's approximation since $m!=\Theta(m^{m+\frac{1}{2}}e^{-m})$. 
\end{proof}

\section{Applications of our  Framework}
\label{sec:applications}
In this section, we apply our elevated flow relaxation framework to extend the previous approximation guarantees in a variety of other settings, providing the first finite approximation guarantees for these settings. In Section \ref{ssec:rerouting}, we show that our framework can also extend to rebalancing controls other than pricing such as supply redirection (empty-vehicle rebalancing) and demand redirection (matching decisions). Next, in Section~\ref{ssec:point_pricing}, we discuss constrained optimization settings such as the case when prices need to come from a discrete set. 
Finally, in Section \ref{ssec:multi_objective}, we show approximation guarantees for multi-objective settings such as Ramsey pricing for which we obtain bicriterion approximations. Although we  describe the above results in the absence of travel-times, for ease of presentation, they mostly extend directly via an added rate-limiting constraint (Section~\ref{ssec:reb_delays}). 
The results in this section are based on the bounds in Sections \ref{sec:pricing} and \ref{ssec:delay}, combined with the recognition that the controls we derive still induce Gordon-Newell/BCMP networks.

\subsection{Rebalancing Controls Beyond Pricing}
\label{ssec:rerouting}
Pricing is just one of several control levers in shared vehicle systems for balancing supply and demand; we now investigate two other levers, which we refer to as \emph{supply redirection} and \emph{demand redirection}, and show how they fit into our approximation framework. In supply redirection we make a decision at the end of every trip on whether the unit remains at the destination of the trip or moves elsewhere whilst incurring a cost. This is of particular importance in systems in which the platform has great control over the units, e.g., in a system of autonomous vehicles or in schemes like Lyft's personalized power zones \cite{LyftPPZ} (assuming drivers' compliance). In demand redirection, we redirect passengers arriving at a node to take units from nearby nodes. In practice, this would be achieved by pulling units from nearby nodes; for example in ridesharing services, the platform can dispatch a driver from a nearby area.

Before delving into the details of these additional controls, it is worth emphasizing two points. First, we note that different controls give rise to fundamentally different problems/objectives, and though we use the same methodology (based on the elevated flow relaxation) to get approximation guarantees for these controls, note that the optimal performance (and hence that of our policies) under different controls may be very different. The simplest example for that is given by the subgraphs of Figure \ref{fig:nonconcave} induced on just B and C: Algorithm \ref{alg:demand_circulation_reduction_general_pricing} would return a solution that allows only $O(\epsilon)$ demand to be served, whereas our results in this section (see Algorithm \ref{alg:demand_circulation_reduction_routing}) may allow for vehicles arriving at $C$ to be redirected to $B$ (depending on the cost of redirection) and can thereby serve $O(1)$ customers. Second, the examples here serve to illustrate that our technique apply to any controls that result in a linear transformation (i.e., scaling, splitting and redirection) of demand (i.e., passenger) and supply (i.e., unit) flows.

\paragraph{Supply Redirection.} 
We consider a state-dependent policy $\r(\X)$ which, for each trip ending at a node $i$, chooses to redirect the unit to some other node $j$ (leading to state $\X-e_i+e_j$), else allows the unit to stay at $i$. For a state-independent policy, let $r_{ij}\in[0,1]$ be the probability that an arriving unit at $i$ is redirected to $j$. We assume that each redirection from $i$ to $j$ has associated cost $c_{ij}$, and that units arriving empty (redirected) are not redirected again.

With $m$ units, fixed (state-dependent) quantiles $\q$, and a (state-dependent) redirection policy $\r$, we observe a rate $f_{ij,m}(\q,\r)$ of customers traveling from $i$ to $j$, and a rate of redirected vehicles $z_{ij,m}(\q,\r)$ from $i$ to $j$, i.e. trips with destination $i$ which are redirected to $j$. For a state-independent policy, since each unit arriving at $i$ is redirected to $j$ with probability $r_{ij}$, it holds that
$$
z_{ij,m}(\q,\r)=r_{ij}\sum_k f_{ki,m}(\q,\r).
$$
Similarly to the correspondence between $q_{ij}$ and $f_{ij,m}$, we observe a correspondence between $r_{ij}$ and $z_{ij,m}$, wherein the former are the controls and induce the latter in the objective via the steady-state dynamics. As a result, the objective for state-independent policies can  be written as 
$$
\obj_m(\q,\r)=\sum_{i,j} f_{ij,m}(\q,\r) I_{ij}(q_{ij}) - c_{ij}z_{ij,m}(\q,\r).$$ 
In order to define the constraints of the elevated flow relaxation, we write (as in Section \ref{sec:pricing}) $\widehat{q}_{ij} = f_{ij,m}(\q,\r)/\phi_{ij}$ and $\widehat{z}_{ij} = z_{ij,m}(\q,\r)$. We can now write the following relaxed flow polytope:
$$(1)\;\;\widehat{q}_{ij}\in[0,1],\qquad
(2)\;\;\sum_k \prn*{\phi_{ki}\widehat{q}_{ki} + \widehat{z}_{ki}}  = \sum_{j}\prn*{ \phi_{ij}\widehat{q}_{ij}+\widehat{z}_{ij}},\qquad
(3)\;\;\sum_k \widehat{z}_{ik} \leq \sum_j \phi_{ji}\widehat{q}_{ji}\;\forall\;i.$$

The first constraint is demand bounding, exactly as explained in Section \ref{sec:pricing}. The second is a variant of the supply circulation in Section \ref{sec:pricing} to incorporate redirected vehicles (the proof is the same as in Proposition \ref{prop:supply_circ}). Finally, the third reflects that only units that are dropping off customers at a node, but not empty ones, can be redirected. Note that these constraints hold for any state-dependent policy as any policy induces such rates $f_{ij,m}$ and $z_{ij,m}$.

Using the reward curves $R_{ij}(\cdot)$ defined in Section \ref{sec:pricing}, we obtain an upper bound $\widehat{\obj}(\q,\r)$ on our desired objective via the elevated flow relaxation with the above constraints; through this, we obtain prices and redirection probabilities in Algorithm \ref{alg:demand_circulation_reduction_routing}. Note that the redirection probabilities $r_{ij}$ returned by the algorithm correspond to the rate of redirected units $z_{ij}$ returned by the relaxation over the total incoming rate of (non-empty) units at node $i$.
\begin{algorithm}[!h]
\caption{The Elevated Flow Relaxation Program with Supply Redirection}
\label{alg:demand_circulation_reduction_routing}
\begin{algorithmic}[1]
\REQUIRE arrival rates $\phi_{ij}$, value distributions $F_{ij}$,  reward curves~$R_{ij}$, rerouting costs $c_{ij}$. \\ 
\STATE Find $\{\widetilde{q}_{ij},\widetilde{z}_{ij}\}$ that solve the following relaxation:
\[
\begin{array}{rcll}
\max_{\q,\z} &&  \sum_{i,j} \prn*{\phi_{ij}R_{ij}({q}_{ij})-c_{ij}{z}_{ij}} &\\
\sum_k \prn*{\phi_{ki}{q}_{ki} + {z}_{ki}}  & = & \sum_{j}\prn*{ \phi_{ij}{q}_{ij} + {z}_{ij}} & \forall i\\
\sum_k {z}_{ik} & \leq & \sum_j \phi_{ji}{q}_{ji} &\forall i\\
{q}_{ij} & \in & [0,1] & \forall i,j
\end{array}
\]
\STATE Output \emph{state-independent} prices $\widetilde{p}_{ij}= F_{ij}^{-1}(1-\widetilde{q}_{ij})$ and the corresponding quantiles $\widetilde{q}_{ij}$ as well as redirection probabilities $\widetilde{r}_{ij} = \widetilde{z}_{ij}/\sum_{k}\phi_{ki}$.
\end{algorithmic}
\end{algorithm}

For the approximation guarantee, we need to first provide an equivalent to Lemma~\ref{lem:equal_availabilities}. Its role in our three-step process can be replaced by the following lemma.

\begin{lemma}\label{lem:reb_availability}\vspace{0.1in}
With $\widetilde{\q}$ and $\widetilde{\r}$ as returned by Algorithm 
\ref{alg:demand_circulation_reduction_routing}, all availabilities are equal to one in the infinite-unit system.
\end{lemma}
\begin{proof}
Our approach is based on the following intuition: We consider a closed queueing network with the same transition probabilities between states as the one resulting from $\widetilde{\q}$ and $\widetilde{\r}$. In our hypothetical network, quantiles are all one, there is no redirection, and the balanced demand property holds. Since the hypothetical network does not have redirection and satisfies the balanced demand property, Lemma \ref{lem:equal_availabilities} implies that there the availabilities at all nodes are equal. However, the two networks have the same transition probabilities so they also have the same steady-state distribution. As a result, in the original network all availabilities are also equal and thus, equal to 1 in the infinite-unit limit.
To prove the lemma we need to (i) define the demand in the hypothetical network such that transitions occur at the same rate as in the network with $\widetilde{\q}$ and $\widetilde{\r}$ and (ii) show that the demand in the hypothetical network fulfills the balanced demand property. We define the demand in the hypothetical network as
$$
\bar{\phi}_{ij} = \phi_{ij}\widetilde{q}_{ij}(1-\sum_k \widetilde{r}_{jk}) + \sum_k \phi_{ik}\widetilde{q}_{ik}\widetilde{r}_{kj}.
$$
Observe that transitions occur at the same rate in this network as in the one with $\widetilde{\q}$ and $\widetilde{\r}$. 
Since quantiles are equal to 1, the balanced demand property requires that $\sum_j \bar{\phi}_{ij}=\sum_k\bar{\phi}_{ki}$. To show this equality, notice first that the demand at node $i$ is
$$
\sum_{j} \bar{\phi}_{ij}
= \sum_{j} \phi_{ij}\widetilde{q}_{ij} - \sum_j \phi_{ij}\widetilde{q}_{ij}\prn*{\sum_{k} \widetilde{r}_{jk}} +\sum_j\sum_{k} \phi_{ik}\widetilde{q}_{ik}\widetilde{r}_{kj}=\sum_j \phi_{ij}\widetilde{q}_{ij}.
$$
On the other hand, due to the definition of $\bar{\phi}_{ij}$ (first equality), the definition of $\widetilde{r}_{ij}$  in Algorithm \ref{alg:demand_circulation_reduction_routing} (third equality), and the supply circulation constraint in Algorithm \ref{alg:demand_circulation_reduction_routing} (last equality), the demand of customers traveling to $i$ is
\begin{align*}   \sum_{k}\bar{\phi}_{ki}&=\sum_{k}\phi_{ki}\widetilde{q}_{ki}- \sum_{j,k} \phi_{ki}\widetilde{q}_{ki}\widetilde{r}_{ij}+\sum_{j,k} \phi_{kj}\widetilde{q}_{kj}\widetilde{r}_{ji}\\
    &=\sum_{k}\phi_{ki}\widetilde{q}_{ki}- \sum_{j}\widetilde{r}_{ij} \big(\sum_k\phi_{ki}\widetilde{q}_{ki}\big)+
   \sum_{j} \widetilde{r}_{ji} \big(\sum_k\phi_{kj}\widetilde{q}_{kj}\big)
    \\  &=\sum_{k}\phi_{ki}\widetilde{q}_{ki}- \sum_{j}\frac{\widetilde{z}_{ij}}{\sum_k \phi_{ki}\widetilde{q}_{ki}} \big(\sum_k\phi_{ki}\widetilde{q}_{ki}\big)+ \sum_{j} \frac{\widetilde{z}_{ji}}{\sum_k \phi_{kj} \widetilde{q}_{kj}} \big(\sum_k\phi_{kj}\widetilde{q}_{kj}\big)\\ 
    &= \sum_{k}\phi_{ki}\widetilde{q}_{ki} + \sum_{j} \prn*{\widetilde{z}_{ji} - \widetilde{z}}_{ij} =\sum_j \phi_{ij}\widetilde{q}_{ij}. 
\end{align*}
Thus, the balanced demand property is fulfilled by the demand in the hypothetical network and the proof is complete.
\end{proof}
We can now apply our framework to provide the approximation guarantee for this setting.
\begin{theorem}[approximation guarantee for supply redirection]\label{thm:rerouting}\vspace{0.1in}
Consider any objective function $\obj_m$ for the $m$-unit system with concave reward curves $R_{ij}(\cdot)$. Let $\widetilde{\q}$ and $\widetilde{\r}$ be the pricing and redirection policies returned by Algorithm~\ref{alg:demand_circulation_reduction_routing}, and $\opt_m$ be the value of the objective function for the optimal state-dependent pricing and redirection policies in the $m$-unit system. Then
$$
\obj_m(\widetilde{\q},\widetilde{\r}) \geq \frac{m}{m+n-1}\opt_m.
$$
\end{theorem}

\begin{proof}
The proof closely resembles that of Theorem \ref{thm:pricing}. As before, we show the inequality through three intermediate steps: (i) $\widehat{\obj}(\widetilde{\q},\widetilde{\r})\geq \opt_m$, (ii) $\widehat{\obj}(\widetilde{\q},\widetilde{\r})= \obj_\infty(\widetilde{\q},\widetilde{\r})$, and (iii) $\obj_m(\widetilde{\q},\widetilde{\r}) \geq \frac{m}{m+n-1}\obj_\infty(\widetilde{\p},\widetilde{\r})$. The proof of the first inequality is the same as in Lemma \ref{lem:state_dependent}, with the relaxation defined in Algorithm \ref{alg:demand_circulation_reduction_general_pricing} replaced by the relaxation defined in Algorithm \ref{alg:demand_circulation_reduction_routing}. The second step relies on Lemma \ref{lem:reb_availability}, which uses Lemma \ref{lem:equal_availabilities} to prove that in the infinite-unit system all availabilities are 1. Based on this claim, similarly to the proof of Lemma~\ref{lem:infinite}, we observe that the flow of customers on each edge is $\phi_{ij}\widetilde{q}_{ij}$. The definition of the redirection probabilities in Algorithm~\ref{alg:demand_circulation_reduction_routing} then immediately implies that $z_{ij,\infty}(\widetilde{\p},\widetilde{\r})=\widetilde{z}_{ij}$, i.e. the flow of redirected units from $i$ to $j$ is also equal to the value of $\widetilde{z}_{ij}$ in the solution of the relaxation. Finally, for the third step, we apply the same proof as in Lemma \ref{lem:finite} with just one small modification in the auxiliary Lemma~\ref{lem:max_availability}. In that lemma $B_i(\q)$ denotes the contribution of node $i$ per unit of time in which a unit is present at $i$. Previously, this just captured rides leaving node $i$. Now, we also charge~$B_i(\q)$ for the cost incurred through the possible redirection of vehicles traveling from $i$ to $j$ that are redirected to~$k$. Setting~$B_i(\q)$ equal to $\sum_j \phi_{ij}q_{ij} \prn*{I_{ij}\prn*{q_{ij}} - \sum_k \widetilde{r}_{jk}c_{jk}}$ formalizes this charging argument -- the remainder of the proof of Lemma \ref{lem:finite} is equivalent to what is presented in Section \ref{sec:pricing}. This concludes the proof of the theorem.
\end{proof}

In Appendix \ref{app:no_prices} we show that the results obtained in this section continue to hold in settings, in which repositioning is allowed, but pricing is not. This is exactly the setting studied in the concurrent work of \cite{braverman2016empty}. As the next example shows, in such scenarios, the optimal solution might not have the balanced demand property. Nevertheless, our techniques yield $m/(m+n-1)$ approximations. 

\textbf{Example. }
Consider two stations $A$ and $B$ with $\phi_{AB}=2,\phi_{BA}=1$; suppose we want to maximize $\sum_{ij}\f_{ij,m}(\r)-c_{ij}z_{ij}$ where $c_{AB}=c_{BA}=3$. Clearly, the optimal solution (for any $m$) is to not rebalance at all: each time a vehicle is rebalanced, it adds at most one ride to the objective while incurring a cost of 3. Yet, without rebalancing, we have $A_{A,\infty}=\frac{1}{2}, A_{B,\infty}=1$, i.e., the balanced demand property does not hold.

\paragraph{Demand Redirection.}
For the control defined in this section, we assume that there exists a graph $G=(V,E)$ on the set of nodes with edges between nodes that are so close that a customer arriving at one node can be served through a vehicle at an adjacent node. We consider a state-dependent policy $\r(\X)$ which, for each customer arriving at node $i$ willing to pay the price quoted, decides from which node in $\crl*{i}\cup\crl*{j:(i,j)\in E}$, the customer is served. With~$m$ units, fixed quantiles $\q(\X)$, and a fixed matching policy $\r(\X)$, we observe a rate $f_{ij,m}(\q,\r)$ of customers arriving at $i$ that travel to $j$, potentially after being matched to a unit at~$k$, and a rate $z_{ik,m}(\q,\r)$ of customers that arrived to travel from $i$ but have been matched to a unit at~$k$. We can write the objective in this setting as $\obj_m(\q,\r)=\sum_{i,j}f_{ij,m}(\q,\r)I_{ij}(\q)$. We again write $\widehat{q}_{ij} = f_{ij,m}(\q,\r)/\phi_{ij}$ and $\widehat{z}_{ij} = z_{ij,m}(\q,\r)$ to define the following relaxed flow polytope:
$$(1)\;\;\widehat{q}_{ij}\in[0,1],\qquad
(2)\;\;\sum_k \widehat{q}_{ki}\phi_{ki} + \widehat{z}_{ik} =  \sum_j \widehat{q}_{ij}\phi_{ij} + \widehat{z}_{ji}\;\forall\;i,\qquad
(3)\;\;\sum_k \widehat{z}_{ki} \leq \sum_j \widehat{q}_{ji}\phi_{ji}\;\forall\;i.$$
The first constraint is again demand bounding. The second is a variant of the
supply circulation to incorporate matchings to nearby nodes. In particular, the left hand side accounts for the total number of units arriving at node $i$, which equals all users arriving at $i$ together with all units arriving due to matching from nearby nodes $k$. Similarly, the right hand side accounts for the total number of units leaving $i$, which are the users leaving from $i$ together with users from other nodes $j$ that use supply at $i$. Finally, the third ensures that customers are matched only to units arriving at nearby nodes. Maximizing the elevated objectives over these constraints again yields a $m/(m+n-1)$ approximation algorithm. We omit the proof, because of its similarity to the one of Theorem \ref{thm:rerouting}.

\begin{algorithm}[!h]
\caption{The Elevated Flow Relaxation Program With Matching}
\label{alg:demand_circulation_reduction_matching}
\begin{algorithmic}[1]
\REQUIRE arrival rates $\phi_{ij}$, value distributions $F_{ij}$,  reward-curves $R_{ij}$, edges $E$. \\ 
\STATE Find $\{\widetilde{q}_{ij},\widetilde{z}_{ij}\}$ that solve the following relaxation:
\[
\begin{array}{rcll}
\max_{\q,\z} &&  \sum_{i,j} \phi_{ij}R_{ij}({q}_{ij})&\\
\sum_k \prn*{\phi_{ki}{q}_{ki} +{z}_{ik}}  & = & \sum_{j}\prn*{ \phi_{ij}{q}_{ij} +{z}_{ji}} & \forall i \in [n] \\
\sum_k {z}_{ki} & \leq & \sum_j \phi_{ji}{q}_{ji} &\forall i \in[n]\\
{q}_{ij} & \in & [0,1] & \forall i,j\in[n]\\
 {z}_{ij} & = & 0 & \forall (i,j)\not\in E
\end{array}
\]
\STATE Output \emph{state-independent} prices $\widetilde{p}_{ij}= F_{ij}^{-1}(1-\widetilde{q}_{ij})$ and the corresponding quantiles $\widetilde{q}_{ij}$
as well as matching probabilities $\widetilde{r}_{ij} = \widetilde{z}_{ij}/\sum_{k}\phi_{ik}\widetilde{q}_{ik}$
\end{algorithmic}
\end{algorithm}

\begin{theorem}[approximation guarantee for demand redirection]\vspace{0.1in}
\label{thm:matching}
Consider any objective function $\obj_m$ for the $m$-unit system with concave reward curves $R_{ij}(\cdot)$. Let $\widetilde{\q}$ and $\widetilde{\r}$ be the pricing and matching policies returned by Algorithm~\ref{alg:demand_circulation_reduction_matching}, and $\opt_m$ be the objective of the optimal state-dependent policies in the $m$-unit system. Then
$$
\obj_m(\widetilde{\q},\widetilde{\r}) \geq \frac{m}{m+n-1}\opt_m.
$$
\end{theorem}

\subsection{Constrained Point Pricing}
\label{ssec:point_pricing}
In this section, we focus on a special case of the vanilla pricing problem wherein the platform is only allowed to set point prices, i.e. prices based on the origin node. The motivation to study such a model is given by pricing policies like Uber's early surge pricing that multiplied a base fare for a trip by an destination-independent \emph{surge factor}. We assume that customers' sensitivity to surge pricing is independent of their destination meaning that the value distributions of all customers arriving at a node are identical. We show that under this assumption, the elevated flow relaxation for the point pricing problem can be solved through just one eigenvalue computation for throughput/welfare; for the problem of maximizing revenue, solving the elevated flow relaxation requires solving an additional one-dimensional concave maximization problem. Given that Uber's early surge pricing policy provided a surge factor that was required to be a multiple of $\frac{1}{10}$, we then consider the additional constraint that prices are only allowed to come from a discrete price set. Interestingly, as the example below shows, with this restriction there may not even be a \emph{feasible} solution that satisfies the balanced demand property. Nevertheless, we show that our results provide finite guarantees, the quality of which depends on how well the discrete price set represents each part of the distribution. We provide the algorithms, results, and proofs for this setting in Appendix~\ref{app:point_pricing}.

\textbf{Example. }
Consider two stations $A$ and $B$ with $\phi_{AB}=\pi, \phi_{BA}=1$; suppose value distributions in both directions are $U[0,1]$, and we are trying to maximize revenue with quantiles (prices) restricted, in either direction, to the set  $\{k/100: k\in\{0,\ldots,99\}\}$. Then there exists no feasible quantiles $q_{AB},q_{BA}$ such that $\phi_{AB}q_{AB}=\phi_{BA}q_{BA}$, i.e., pricing cannot enforce balanced demand in the network.

\subsection{Pricing in Multi-Objective Settings}\label{ssec:multi_objective}
In this part, we apply our framework to multi-objective settings where we wish to maximize one objective function subject to a lower bound on another one, as is the case in Ramsey pricing \cite{ramsey1927contribution}; the corresponding elevated flow relaxation program is decribed in Algorithm~\ref{alg:multi_objective}. Formally, the problem is as follows: we are given a $m$-unit system, a requirement $c\geq 0$, and objectives $\Phi_m(\cdot)$ and $\Psi_m(\cdot)$; the goal is to maximize $\Phi_m(\q)$ subject to $\Psi_m(\q)\geq c$. We again assume that both objectives can be decomposed into per-ride rewards with associated concave reward curves $\crl*{R_{ij}^{\Psi}}$ and $\crl*{R_{ij}^{\Phi}}$.

Similarly to Equation \eqref{eq:elevated}, we first elevate both objectives to obtain $\widehat{\Phi}(\widehat{\q})=\sum_{i,j}\phi_{ij}R^{\Phi}_{ij}(\widehat{q}_{ij})$ and $\widehat{\Psi}(\widehat{\q})=\sum_{i,j}\phi_{ij}R^{\Psi}_{ij}(\widehat{q}_{ij})$. Since per-ride rewards are non-increasing on the quantiles, this can only increase the values of the objectives.
We then impose the supply circulation and demand bounding constraints to create the flow polytope constraints. This mathematical program  (Algorithm \ref{alg:multi_objective}) is the elevated flow relaxation program for our multi-objective setting; we argue below that this is indeed a relaxation. It can be efficiently optimized since the objective is concave while the polytope is convex: the convex combination of any two feasible quantiles is feasible since $\widehat{\Psi}(\cdot)$ is concave.
\begin{algorithm}[!h]
\caption{The Elevated Flow Relaxation Program with Multi-objective Pricing}
\label{alg:multi_objective}
\begin{algorithmic}[1]
\REQUIRE arrival rates $\phi_{ij}$, value distributions $F_{ij}$,  reward curves~$R^{\Phi}_{ij}$ and $R^{\Psi}_{ij}$, requirement ~$c$.
\STATE Find $\crl*{\widetilde{q}}_{ij}$ that solve the following relaxation:
$$
\begin{array}{rcll}
\max_\q &&  \widehat{\Phi}(\q) &\\
\sum_{k} \phi_{ki}q_{ki}  & = & \sum_{j} \phi_{ij}q_{ij} & \forall\;i\\
q_{ij} & \in & [0,1] & \forall\;i,j.\\
\widehat{\Psi}(\q) & \geq & c &
\end{array}
$$
\STATE Output \emph{state-independent} prices $\widetilde{p}_{ij}= F_{ij}^{-1}(1-\widetilde{q}_{ij})$ and the corresponding quantiles $\widetilde{q}_{ij}$.
\end{algorithmic}
\end{algorithm}
\begin{theorem}[approximation guarantee for multi-objective pricing]\label{thm:multi_obj_infinity}\vspace{0.1in}
Let $\Phi_m$ and $\Psi_m$ be objectives for the $m$-unit system with concave reward curves. Then the solution $\widetilde{\q}$ returned by Algorithm~\ref{alg:multi_objective} is a $(\gamma, \gamma)$ bicriterion approximation for the multi-objective pricing problem where $\gamma=\nicefrac{m}{(m+n-1)}$, i.e. $\Phi_m(\widetilde{\q})\geq \gamma\opt_m
$
and $\Psi_m(\widetilde{\q})\geq \gamma \cdot c$.
\end{theorem}
\begin{proof}
Let $\widehat{\q}$ denote the optimal solution of an auxiliary program where we only elevate the objective $\Phi$ (but not $\Psi$), i.e. we maximize $\widehat{\Phi}(\cdot)$ subject to $\Psi_m(\cdot)\geq c$ as well as the balanced demand property and demand bounding constraints. Moreover, let $\q^{\star}$ denote the optimal state-dependent solution of the original (non-elevated) problem. Then, for the first guarantee, we have:
\begin{align*}
    \Phi_m(\widetilde{\q})\geq \gamma \widehat{\Phi}(\widetilde{\q})\geq \gamma \widehat{\Phi}(\widehat{\q})
    \geq \gamma\Phi(\q^{\star})= \gamma \opt_m
\end{align*}
The first inequality is a simple application of Lemmas \ref{lem:infinite} and \ref{lem:finite}, since $\widetilde{\q}$ satisfies balanced demand and we can thus apply the lemmas to the objective $\Phi_m(\cdot)$. The second inequality holds since any solution of the auxiliary program is a feasible solution of the elevated flow relaxation. In particular, since the elevated objective $\widehat{\Phi}(\cdot)$ is pointwise no less than the original objective $\Phi_m(\cdot)$, the corresponding constraint in the auxiliary program is tighter. The last inequality holds by the same argument as in Lemma~\ref{lem:state_dependent}.

Regarding the second guarantee, we have:
\begin{align*}
    \Psi_m(\widetilde{\q})\geq \gamma \widehat{\Psi}(\widetilde{\q})\geq \gamma c
\end{align*}
The first inequality is again an application of
Lemmas~\ref{lem:infinite} and \ref{lem:finite}, while the second holds since $\widetilde{\q}$ is a feasible solution of the elevated flow relaxation and therefore satisfies its last constraint.
\end{proof}

We note that the same approach yields multicriterion approximation algorithms for settings in which more than one constraint of the form $\Psi_m(\cdot)\geq c$ is given. This approach also extends when constraints are given that limit the total amount of repositioning, e.g., due to city governments limiting congestion. Then, with the notation of Section \ref{ssec:rerouting}, we may add a constraint $\sum_{i,j}c_{ij}z_{ij}\leq c'$ to the flow relaxation; by the same reasoning as before, we maintain an upper bound, this is asymptotically achieved, and if the constraint is fulfilled in the infinite-unit system, it is also fulfilled in the finite-unit system (while the objective remains within a factor of $1+(n-1)/(m+n-1)$).

\subsection{Extending Results to Travel Times}\label{ssec:reb_delays}
The results in this section extend to the rate-limiting relaxations in Section \ref{ssec:delay} to provide guarantees when travel-times are included into the models. When pricing is permitted, this is immediate. For example, to extend Algorithm \ref{alg:demand_circulation_reduction_routing} we can again solve the same relaxation with an added rate-limiting constraint $\sum_{i,j} (\phi_{ij}\hat{q}_{ij} +\hat{z}_{ij})\tau_{ij} \leq m$ but then return $\widetilde{q}_{ij} =  q_{ij}\cdot\prn*{1 - \varepsilon_m}$. Since any policy must fulfill the rate-limiting constraint the resulting mathematical program still provides an upper bound on any optimal objective. The rest of the argument is then exactly equivalent to that of Lemma \ref{lem:Am_w_delay} and Theorem~\ref{thm:pricing_w_delays}. For the case that pricing is not permitted, adapting results to travel-times requires a little more work, as returning $\widetilde{q}_{ij}$ is not feasible; we discuss how to obtain finite guarantees in that case,  including ones that generalize \cite{braverman2016empty}, in the second half of Appendix \ref{app:no_prices}.

\section{Conclusions}\label{sec:conclusions}

We have studied pricing and optimization in shared vehicle systems for various objectives. Our work parallels existing work through our use of a closed-queueing network model to capture network externalities. It distinguishes itself, however, through the unified approach that handles a wide range of objectives, controls and system constraints, while providing rigorous finite-system guarantees for all, including when there are travel-times, a setting for which there had not been any finite guarantees before. Furthermore, our results hold equally true when networks are not balanced, which requires a bound, relative to the optimum, that does not only depend on factoring out the availability term. Given the differences in regulations, product offerings, and operational capabilities, across the shared mobility space, this is an important generalization. Thus,
our approach unifies and extends several results from the literature. 

Beyond our algorithmic results, our results also indicate a high-level design feature for real-time decision-making in shared vehicle systems: enforcing conservation laws (flow conservation, Little's law) in expectation yields near-optimal policies in regimes of interest. Although our model does not capture all features of on-demand transportation systems (e.g., endogenous driver decisions and pooled rides in ride-sharing, limited docking places in bike-sharing, etc.) we believe that this core principle can prove useful in these more general settings. Moreover, by establishing finite guarantees for a wide class of settings, we provide formal justification for the widespread use of fluid models for the analysis of shared vehicle systems.

From a technical viewpoint, our main contribution, the elevated flow relaxation, has the potential to apply to other settings. At a high-level, our approach can be viewed as providing a way to couple a fluid model and an appropriate product-form distribution for a given stochastic system. Given the widespread use of fluid limits in the queuing theory literature, as well as a growing knowledge of product-form characterizations, our framework may prove useful for finite instances in these settings as well. The key to such applications involves understanding the associated partition function (in our case: $G_m$), which is known to be challenging for other models.

More specific to our own model, it would be interesting to study how our framework can be extended to constrained settings beyond multi-objective and discrete prices. For instance, in recent events, Uber was exposed to bad publicity when turning off surge pricing for trips originating at JFK airport. While the details of these events are not mathematical in nature, it demonstrates the significance of studying settings where prices (in some locations) are bounded above; from a technical perspective, this is a challenging problem because it complicates the derivation of a flow relaxation that appropriately couples with the stochastic system. Additionally, our pricing policies do not impose triangle inequality, potentially creating incentives for customers to reach their destination via an extra stop. Addressing such strategic considerations opens up an intriguing avenue for future research.

Finally, although our work suggests that state-independent prices have strong performance, this is under the steady-state assumption with complete knowledge of the system parameters. Relaxing either of these assumptions is a compelling extension of our work.

\section*{Acknowledgments.}
The authors would like to thank David Shmoys and \'{E}va Tardos for their support and their fruitful feedback throughout this work. The authors also thank Anton Braverman and Jim Dai for many useful conversations. Finally, our results and their presentation have been significantly improved by the excellent feedback from  anonymous reviewers. This work was supported, in part, by NSF grants CCF-1526067, CMMI-1537394, CCF-1563714, CCF-1522054, CCF-1740822, CNS-1955997, DMS-1839346, and ECCS-1847393, ARL grant W911NF-17-1-0094, ONR grant N00014-08-1-0031,  a Google PhD Fellowship, and a Google Faculty Award.

\bibliographystyle{alpha}
\bibliography{bib1}

\begin{thebibliography}{BCMP75}

\bibitem[Ade07]{adelman2007price}
Daniel Adelman.
\newblock Price-directed control of a closed logistics queueing network.
\newblock {\em Operations Research}, 55(6):1022--1038, 2007.

\bibitem[ALM18]{afeche2018ride}
Philipp Af{\`e}che, Zhe Liu, and Costis Maglaras.
\newblock Ride-hailing networks with strategic drivers: The impact of platform
  control capabilities on performance.
\newblock 2018.

\bibitem[BBC20]{balseiro2020dynamic}
Santiago~R Balseiro, David~B Brown, and Chen Chen.
\newblock Dynamic pricing of relocating resources in large networks.
\newblock {\em Management Science}, 2020.

\bibitem[BCL21]{castro2018surge}
Omar Besbes, Francisco Castro, and Ilan Lobel.
\newblock Surge pricing and its spatial supply response.
\newblock {\em Management Science}, 67(3):1350--1367, 2021.

\bibitem[BCMP75]{baskett1975}
Forest Baskett, K~Mani Chandy, Richard~R Muntz, and Fernando~G Palacios.
\newblock Open, closed, and mixed networks of queues with different classes of
  customers.
\newblock {\em Journal of the ACM (JACM)}, 22(2):248--260, 1975.

\bibitem[BCS19]{bimpikis2016spatial}
Kostas Bimpikis, Ozan Candogan, and Daniela Saban.
\newblock Spatial pricing in ride-sharing networks.
\newblock {\em Operations Research}, 67(3):744--769, 2019.

\bibitem[BDLY19]{braverman2016empty}
Anton Braverman, Jim~G Dai, Xin Liu, and Lei Ying.
\newblock Empty-car routing in ridesharing systems.
\newblock {\em Operations Research}, 67(5):1437--1452, 2019.

\bibitem[BES19]{besbes2019static}
Omar Besbes, Adam~N Elmachtoub, and Yunjie Sun.
\newblock Static pricing: Universal guarantees for reusable resources.
\newblock In {\em Proceedings of the 2019 ACM Conference on Economics and
  Computation}, pages 393--394, 2019.

\bibitem[BES20]{besbes2020pricing}
Omar Besbes, Adam~N Elmachtoub, and Yunjie Sun.
\newblock Pricing analytics for rotable spare parts.
\newblock {\em INFORMS Journal on Applied Analytics}, 50(5):313--324, 2020.

\bibitem[BJR15]{banerjee2015pricing}
Siddhartha Banerjee, Ramesh Johari, and Carlos Riquelme.
\newblock Pricing in ride-sharing platforms: A queueing-theoretic approach.
\newblock In {\em Proceedings of the Sixteenth ACM Conference on Economics and
  Computation}, pages 639--639. ACM, 2015.

\bibitem[BK09]{bulow2009sellers}
Jeremy Bulow and Paul Klemperer.
\newblock Why do sellers (usually) prefer auctions?
\newblock {\em American Economic Review}, 99(4):1544--75, 2009.

\bibitem[BKM13]{Brooks2013}
James~D. Brooks, Koushik Kar, and David Mendonça.
\newblock Dynamic allocation of entities in closed queueing networks: An
  application to debris removal.
\newblock In {\em Proceedings of the 2013 IEEE International Conference on
  Technologies for Homeland Security}, pages 504--510, 2013.

\bibitem[Buz73]{Buzen}
Jeffrey~P Buzen.
\newblock Computational algorithms for closed queueing networks with
  exponential servers.
\newblock {\em Communications of the ACM}, 16(9):527--531, 1973.

\bibitem[CFS18]{chung2018bike}
Hangil Chung, Daniel Freund, and David~B Shmoys.
\newblock Bike angels: An analysis of citi bike's incentive program.
\newblock In {\em Proceedings of the 1st ACM SIGCAS Conference on Computing and
  Sustainable Societies}, page~5. ACM, 2018.

\bibitem[CH20]{chen2017pricing}
Yiwei Chen and Ming Hu.
\newblock Pricing and matching with forward-looking buyers and sellers.
\newblock {\em Manufacturing \&amp; Service Operations Management},
  22(4):717--734, 2020.

\bibitem[CKW17]{castillo2017surge}
Juan~Camilo Castillo, Dan Knoepfle, and Glen Weyl.
\newblock Surge pricing solves the wild goose chase.
\newblock In {\em Proceedings of the 2017 ACM Conference on Economics and
  Computation}, pages 241--242. ACM, 2017.

\bibitem[CY13]{chen2013fundamentals}
Hong Chen and David~D Yao.
\newblock {\em Fundamentals of queueing networks: Performance, asymptotics, and
  optimization}, volume~46.
\newblock Springer Science \& Business Media, 2013.

\bibitem[Fai17]{zombiecars}
Marcus Fairs.
\newblock Paris deputy mayor questions london's approach to skyscrapers and
  public space, 2017.

\bibitem[Fol86]{foley1986stationary}
Robert~D Foley.
\newblock Stationary poisson departure processes from non-stationary queues.
\newblock {\em Journal of applied probability}, 23(1):256--260, 1986.

\bibitem[Geo12]{george2012stochastic}
David~K George.
\newblock {\em Stochastic Modeling and Decentralized Control Policies for
  Large-Scale Vehicle Sharing Systems via Closed Queueing Networks}.
\newblock PhD thesis, The Ohio State University, 2012.

\bibitem[GN67]{gordon1967closed}
William~J Gordon and Gordon~F Newell.
\newblock Closed queuing systems with exponential servers.
\newblock {\em Operations research}, 15(2):254--265, 1967.

\bibitem[GVR94]{gallego1994optimal}
Guillermo Gallego and Garrett Van~Ryzin.
\newblock Optimal dynamic pricing of inventories with stochastic demand over
  finite horizons.
\newblock {\em Management science}, 40(8):999--1020, 1994.

\bibitem[GX11]{george2011fleet}
David~K George and Cathy~H Xia.
\newblock Fleet-sizing and service availability for a vehicle rental system via
  closed queueing networks.
\newblock {\em European journal of operational research}, 211(1):198--207,
  2011.

\bibitem[GXS12]{george2012}
David~K George, Cathy~H Xia, and Mark~S Squillante.
\newblock Exact-order asymptotic analysis for closed queueing networks.
\newblock {\em Journal of Applied Probability}, pages 503--520, 2012.

\bibitem[Har13]{hartline2013}
Jason~D Hartline.
\newblock {\em Mechanism design and approximation}.
\newblock 2013.

\bibitem[HMW09]{hampshire2009dynamic}
Robert~C Hampshire, William~A Massey, and Qiong Wang.
\newblock Dynamic pricing to control loss systems with quality of service
  targets.
\newblock {\em Probability in the Engineering and Informational Sciences},
  23(02):357--383, 2009.

\bibitem[HZ19]{hu2017price}
Ming Hu and Yun Zhou.
\newblock Price, wage and fixed commission in on-demand matching.
\newblock {\em Available at SSRN 2949513}, 2019.

\bibitem[Jac63]{jackson1963}
James~R Jackson.
\newblock Jobshop-like queueing systems.
\newblock {\em Management science}, 10(1):131--142, 1963.

\bibitem[Kel11]{kelly2011}
Frank~P Kelly.
\newblock {\em {Reversibility and Stochastic Networks}}.
\newblock Cambridge University Press, 2011.

\bibitem[KQ19]{kanoria2019near}
Yash Kanoria and Pengyu Qian.
\newblock Near optimal control of a ride-hailing platform via mirror
  backpressure.
\newblock {\em arXiv}, pages arXiv--1903, 2019.

\bibitem[KY14]{kelly2014stochastic}
Frank Kelly and Elena Yudovina.
\newblock {\em Stochastic networks}, volume~2.
\newblock Cambridge University Press, 2014.

\bibitem[LR10]{levi2010provably}
Retsef Levi and Ana Radovanovic.
\newblock Provably near-optimal lp-based policies for revenue management in
  systems with reusable resources.
\newblock {\em Operations Research}, 58(2):503--507, 2010.

\bibitem[Lyf18]{LyftPPZ}
Lyft.
\newblock Personalized power zones, 2018.

\bibitem[MFP18]{ma2018spatio}
Hongyao Ma, Fei Fang, and David~C Parkes.
\newblock Spatio-temporal pricing for ridesharing platforms.
\newblock {\em arXiv preprint arXiv:1801.04015}, 2018.

\bibitem[MS14]{milgrom}
Paul Milgrom and Ilya Segal.
\newblock Deferred-acceptance auctions and radio spectrum reallocation.
\newblock In {\em Proceedings of the fifteenth ACM conference on Economics and
  computation}, pages 185--186. ACM, 2014.

\bibitem[New66]{newell1966m}
GF~Newell.
\newblock The $m/g/\infty$ queue.
\newblock {\em SIAM Journal on Applied Mathematics}, 14(1):86--88, 1966.

\bibitem[OFC21]{hao2020eppz}
HaoYi Ong, Daniel Freund, and Davide Crapis.
\newblock Driver positioning and incentive budgeting with an escrow mechanism
  for ridesharing platforms.
\newblock {\em INFORMS Journal on Applied Analytics, forthcoming}, 2021.

\bibitem[{\"O}W20]{ozkan2016dynamic}
Erhun {\"O}zkan and Amy~R Ward.
\newblock Dynamic matching for real-time ride sharing.
\newblock {\em Stochastic Systems}, 10(1):29--70, 2020.

\bibitem[{\"O}zk20]{ozkan2018joint}
Erhun {\"O}zkan.
\newblock Joint pricing and matching in ride-sharing systems.
\newblock {\em European Journal of Operational Research}, 287(3):1149--1160,
  2020.

\bibitem[QBK18]{banerjee2018value}
Pengyu Qian, Siddhartha Banerjee, and Yash Kanoria.
\newblock The value of state dependent control in ridesharing systems.
\newblock {\em arXiv preprint arXiv:1803.04959}, 2018.

\bibitem[Ram27]{ramsey1927contribution}
Frank~P Ramsey.
\newblock A contribution to the theory of taxation.
\newblock {\em The Economic Journal}, 37(145):47--61, 1927.

\bibitem[RL80]{MVA}
Martin Reiser and Stephen~S Lavenberg.
\newblock Mean-value analysis of closed multichain queuing networks.
\newblock {\em Journal of the ACM (JACM)}, 27(2):313--322, 1980.

\bibitem[Ser99]{IntroStochNetworks}
Richard Serfozo.
\newblock Introduction to stochastic networks, 1999.

\bibitem[SSB18]{sejourne2018price}
Thibault S{\'e}journ{\'e}, Samitha Samaranayake, and Siddhartha Banerjee.
\newblock The price of fragmentation in mobility-on-demand services.
\newblock {\em Proceedings of the ACM on Measurement and Analysis of Computing
  Systems}, 2(2):30, 2018.

\bibitem[ST85]{sleator1985amortized}
Daniel~D Sleator and Robert~E Tarjan.
\newblock Amortized efficiency of list update and paging rules.
\newblock {\em Communications of the ACM}, 28(2):202--208, 1985.

\bibitem[Ube18]{ubersnewsurge}
Uber.
\newblock Your questions about the new surge, answered, 2018.

\bibitem[Whi84]{whitt1984}
Ward Whitt.
\newblock Open and closed models for networks of queues.
\newblock {\em AT\&T Bell Laboratories Technical Journal}, 63(9):1911--1979,
  1984.

\bibitem[Whi85]{Whittle1985}
Peter Whittle.
\newblock Scheduling and characterization problems for stochastic networks.
\newblock {\em Journal of the Royal Statistical Society. Series B
  (Methodological)}, pages 407--428, 1985.

\bibitem[WJ14]{Waserhole2014}
Ariel Waserhole and Vincent Jost.
\newblock Pricing in vehicle sharing systems: optimization in queuing networks
  with product forms.
\newblock {\em EURO Journal on Transportation and Logistics}, pages 1--28,
  2014.

\bibitem[ZP16]{zhang2016control}
Rick Zhang and Marco Pavone.
\newblock Control of robotic mobility-on-demand systems: a queueing-theoretical
  perspective.
\newblock {\em The International Journal of Robotics Research},
  35(1-3):186--203, 2016.

\bibitem[ZSEG82]{zahorjan1982balanced}
John Zahorjan, Kenneth~C Sevcik, Derek~L Eager, and Bruce Galler.
\newblock Balanced job bound analysis of queueing networks.
\newblock {\em Communications of the ACM}, 25(2):134--141, 1982.

\bibitem[ZWS20]{zhong2020queueing}
Yueyang Zhong, Zhixi Wan, and Zuo-Jun~Max Shen.
\newblock Queueing versus surge pricing mechanism: Efficiency, equity, and
  consumer welfare.
\newblock {\em Equity, and Consumer Welfare (September 24, 2020)}, 2020.

\end{thebibliography}

\appendix

\section{Irreducibility of the Priced System}\label{sec:irreducibility}

We justify here our assumption from Section \ref{sec:model} that the infinite-unit solutions we obtain induce a connected graph; to do so, we first need to assume that the graph created by edges $(i,j)$ on which $\phi_{ij}>0$ is strongly connected, that is, the directed graph with edge-set $\crl*{(i,j):\phi_{ij}>0}$ contains a path from any node to any other. We then prove that given any solution to the infinite-unit pricing problem, there exists a solution with arbitrarily close objective that also induces a strongly connected graph. For simplicity, we assume that throughput is the objective, yet the extension to other objectives is immediate. Throughout this section we work with the flow $f_{ij,\infty}(\q)$ induced by  demands in the infinite-unit system, but suppress all dependencies on $\infty$ in the notation.

\begin{theorem}[Irreducible Markov Chain]\label{irreducibleObjective}\vspace{0.1in}
Let $\epsilon>0$. Suppose state-independent quantiles $\q$ induce a steady-state rate of units $f_{ij,\infty}(\q)$ on $k$ components; then there exist quantiles $\q'$ that induce $f_{ij,\infty}(\q')$ such that the graph with edge-set $\{(i,j):f_{ij,\infty}(\q')>0\}$ is strongly connected and the throughput with $\q'$ is at least $(1-\epsilon)$ times that of $\q$ in the infinite-unit system.
\end{theorem}
\begin{proof}
Notice first that we may assume without loss of generality that $\q$ fulfills the balanced demand property; indeed, whether an arc has non-zero demand on it is independent of the number of units and by Lemma \ref{lem:state_dependent} there is a solution in the relaxation (with balanced demand) that has an elevated throughput that is no less than $\sum_{i,j}f_{ij}(\q)$. To prove the theorem we repeatedly increase demand on some edges $(i,j)$ with $\phi_{ij}>0$ and $f_{ij}(\q)=0$, but also decrease demand on some edges $(\bar{i},\bar{j})$ with $f_{\bar{i}\bar{j}}(\q)>0$. Equivalently, we increase quantiles $q_{ij}$ and decrease quantiles $q_{\bar{i}\bar{j}}$. To ensure that edges of the second kind do not have their flow reduced by too much relative to $f_{\bar{i}\bar{j}}(\q)$, we set 
$$\delta=\frac{\epsilon}{k}
\times
\min\crl*{
\min_{i,j}
\crl*{f_{ij}(\q):f_{ij}(\q)>0},
\min_{i,j}
\crl*{\phi_{ij}-f_{ij}(\q):\phi_{ij}-f_{ij}(\q)>0}}.$$ 
Whenever we change the demand on an edge, this is done by an additive $\delta$ amount. Reducing flow at most $k$ times to obtain $f_{ij}(\q')$ we guarantee that $\phi_{ij}\geq f_{ij}(\q')\geq (1-\epsilon)f_{ij}(\q)$ holds which implies that the total throughput cannot change by more than a factor $(1-\epsilon)$.

As we assume that our underlying graph with edge-set $\{(i,j):\phi_{ij}>0\}$ is strongly connected, it must be the case that the graph with edge set $\{(i,j):\phi_{ij}q_{ij}>0\}$ contains a minimal sequence of components $C_1,C_2,\ldots,C_d=C_1$, $d>2$, and nodes $u_\ell,v_\ell\in C_\ell$ such that $\phi_{u_{\ell}v_{\ell+1}}>0$, but $f_{u_\ell v_{\ell+1}}=0$. In particular, it being minimal implies that no component other than the first appears repeatedly.
Since each $u_\ell,v_\ell$ are in the same strongly connected component of the graph with edge-set $\{(i,j):f_{ij}(\q)>0\}$, we know that for each $\ell$ there exists a simple path from $u_\ell$ to $v_\ell$ with positive demand on it. We change the quantiles as follows: for all pairs $(u_\ell, v_{\ell+1})$ we increase the quantiles $q_{u_\ell, v_{\ell+1}}$ so that the steady-state rate of units increases by $\delta$ and for each edge along the path from $u_\ell$ to $v_\ell$ we decrease the quantiles so that the steady-state rate of units decreases by $\delta$. At all other edges the quantiles remain unchanged.

We first argue that this again gives rise to balanced demand: Each node along a path within a component has its in-flow and out-flow (of demand) reduced by $\delta$, whereas at the nodes $u_i,v_i$ both the sum of in-flows and the sum of out-flows have remained the same. At all other nodes, nothing is altered. Thus, flow conservation continues to hold. By choice of $\delta$ none of the edge-capacities are violated. Thus, the resulting demand $\{\phi_{ij}q_{ij}\}$ is in the flow relaxation with at most $k-1$ distinct components. Applying this procedure repeatedly, we obtain a single strongly connected component such that the throughput with $\q'$ in the infinite-unit limit (by Lemma \ref{lem:infinite}) is within $(1-\epsilon)$ of the throughput of $f_{ij,\infty}(\q)$. \end{proof}

\section{Concave Reward Curves}\label{sec:concave_reward_curves}
In this section, we investigate conditions under which throughput, social welfare and revenue satisfy the conditions of theorem \ref{thm:pricing}. In particular, we first show that the respective reward curves $R(q)=qI(q)$ are concave. We then prove that the concave reward curves assumption implies the non-increasing (quantiles) per-ride rewards assumption.

\begin{lemma}\label{lem:concave_objective_function}
\vspace{0.1in} Revenue (i) satisfies the assumptions of Theorem \ref{thm:pricing} under regular value distributions, Throughput (ii) and Social Welfare (iii) satisfy the assumptions under any value distribution.
\end{lemma}
\begin{proof}
We drop the subscripts throughout this proof to simplify notation. We begin by considering (i) revenue, for which  the result holds due to the fact that the reward curve is concave if and only if the distribution is regular (see Proposition 3.10 in \cite{hartline2013}). For (ii) throughput, $R(q)=q \cdot I(q)=q$ is a linear function of $q$ for any value distribution and thus concave.

Lastly, for (iii) social welfare, we use the so-called hazard rate $h(y)=\frac{f(y)}{1-F(y)}$ of a distribution $F$ with density $f$. Given $F$, denote by $p(q)$ and $q(p)$ a price as a function of its corresponding quantile and vice-versa. Then, by the definition of hazard rate:
\begin{equation}\label{eq:hazard1}
    q(p)=\exp\prn*{-\int_{0}^{p(q)} h(y)dy}
\end{equation}
Taking logarithms and differentiating, we obtain:
\begin{equation}\label{eq:hazard2}
 -\frac{1}{q(p)}=h\prn*{p(q)}\frac{dp(q)}{dq}
\end{equation}
Hence, as $R(q(p))=q(p)\cdot I\prn*{q(p)}$ and $f(p)=(1-F(p))h(p)=q(p)h(p)$ we have
\begin{eqnarray*}
R(q)=\int_{p(q)}^\infty v f(v)dv=\int_{p(q)}^\infty v h(v)\exp\prn*{-\int_{0}^v h(y) dy} dv
\end{eqnarray*}
The first derivative $\frac{dR(q)}{dq}$ of $R(q)$ is equal to
\begin{equation*}
    -p(q)h(p(q))\exp\prn*{-\int_{y=0}^{p(q)} h(y) dy} \frac{dp(q)}{dq}=\frac{p(q)\exp\prn*{-\int_{y=0}^{p(q)} h(y) dy}}{q(p)}=p(q),
\end{equation*}
where the first equality comes from Equation \eqref{eq:hazard2}, the second from \eqref{eq:hazard1}.

The second derivative is then given by
\begin{equation*}
\frac{d^2R(q)}{dq^2}=\frac{dp(q)}{dq}
=-\frac{1}{qh\prn*{p(q)}}
=-\frac{1-F(p(q))}{f(p(q))q(p)}<0,
\end{equation*}
which concludes the proof of the Lemma.
\end{proof}

\begin{lemma}
\quad If a function $I(\cdot)$ has the property that $qI(q)$ is concave, then $I(\cdot)$ is non-increasing. In particular, if some objective satisfies the concave reward curves assumption, it also satisfies the non-increasing (in quantiles) per-ride rewards assumption.
\end{lemma}
\begin{proof}
Suppose the statement was not true, then there 
must exist $q_1,q_2$ with $0<q_1<q_2$ such that $I(q_1)<I(q_2)$. Let $A=\frac{q_1}{q_2}$. Then
\begin{align*}
   q_1 I(q_2) &=A\cdot q_2 I(q_2)=A\cdot q_2 I(q_2)+\prn*{1-A}\cdot 0\cdot I(0)\\
   &\leq 
   \prn*{A\cdot q_2+(1-A)\cdot 0}I\prn*{A\cdot q_2 +(1-A)\cdot 0}=q_1 I(q_1),
\end{align*}
where the inequality follows from Jensen's inequality since the function $q I(q)$ is a concave function. As $q_1>0$, it follows that $I(q_2)\leq I(q_1)$ and we therefore arrive at a contradiction.
\end{proof}

\section{Appendix on Point Pricing}\label{app:point_pricing}
In this section, we focus on a special case of the pricing problem wherein the platform is only allowed to set point prices, i.e. prices based on the origin node, and the value distributions of all customers arriving at a node are identical (i.e., $q_{ij}=q_i$, and $F_{ij}(\cdot)=F_i(\cdot)$ for all $i,j$). Thinking of the prices as surge multipliers, this reflects how Uber and Lyft used to price until about 2017. It also reflects how companies like Lyft or Lime price free-floating scooters today when riders are offered a price at the origin location, without the company knowing the destination. We provide a simple optimal pricing policy for the infinite-unit system, which involves just one eigenvector computation (for throughput/social welfare) or a concave maximization over a single variable (for revenue).  We then consider the additional constraint that prices are only allowed to come from a discrete price set. Using our infinite-to-finite unit reduction, all our results are then translated back to the finite unit setting.
We emphasize that in the latter restricted settings, there may not be a feasible solution satisfying balanced demand yet nevertheless we obtain guarantees. 

\paragraph{Unrestricted price set.} We begin by providing the elevated flow relaxation program (Algorithm~\ref{alg:point_pricing}) and approximation guarantee (Theorem~\ref{thm:point_pricing}) for unconstrained  point pricing, whose proof follows the same steps as the one of Theorem~\ref{thm:pricing} and we therefore omit it.
\begin{algorithm}[!h]
\caption{The Elevated Flow Relaxation Program with Point Pricing }
\label{alg:point_pricing}
\begin{algorithmic}[1]
\REQUIRE arrival rates $\phi_{ij}$, value distributions $F_{i}$, reward curves~$R_{ij}$.
\STATE Find $\crl*{\widetilde{q}}_{i}$ that solve the following point price relaxation:
$$
\begin{array}{rcll}
\max_\q &&  \sum_{(i,j)} \phi_{ij}R_{ij}({q}_{i}) &\\
\sum_{k} \phi_{ki}{q}_{k}  & = & \sum_{j} \phi_{ij}{q}_{i} & \forall\;i\\
{q}_{i} & \in & [0,1] & \forall\;i.
\end{array}
$$
\STATE Output \emph{state-independent} prices $\widetilde{p}_{i} = F_{i}^{-1}(1-\widetilde{q}_{i})$ with corresponding quantiles $\widetilde{q}_{i}$.
\end{algorithmic}
\end{algorithm}

\begin{theorem}[approximation guarantee for unrestricted point pricing]\quad
\label{thm:point_pricing}
Consider any objective function $\obj_m$ for the $m$-unit system with concave reward curves $R_{ij}(\cdot)$.
Let $\widetilde{\q}$ be the pricing policy returned by Algorithm \ref{alg:point_pricing}, $\opt_m$ be the value of the objective function for the optimal state-dependent point pricing policy in the $m$-unit system. Then
\begin{equation}
\obj_m(\widetilde{\q})\geq \frac{m}{m+n-1}\opt_m
\end{equation}
\end{theorem}

Notice that the optimization problem in Algorithm \ref{alg:point_pricing} has the balanced demand property as a constraint; thus, with the resulting pricing policy, the availability is equal at every node (see Lemma \ref{lem:equal_availabilities}). Recall from Section~\ref{ssec:queuemodel} that the availability  at each node in the infinite-unit system depends on the traffic intensity at that particular node and the maximum traffic intensity among all nodes. Further, the traffic intensity at each node $i$ depends on (i) the $i$th coordinate of the eigenvector $\w(\q)$ of the routing matrix $\{\phi_{ij}q_i/\sum_k\phi_{ik}q_i\}_{i,j\in[n]^2}$, and (ii) the rate of arrivals $\sum_j \phi_{ij}q_{ij}$ at $i$. In particular, $r_i(\q)=w_i(\q)/\sum_j\phi_{ij}q_{i}$. In the setting of point prices however, $\w$ is unaffected by the prices and since the quantiles from Algorithm \ref{alg:point_pricing} fulfill $r_i(\widetilde{\q})=r_j(\widetilde{\q})\forall i,j$ we have that $w_i \sum_k\phi_{jk}\widetilde{q}_j = w_j\sum_k\phi_{ik}\widetilde{q}_i$ for all $i,j$. Substituting in the optimization problem for every $j$
\begin{center}
$
\widetilde{q}_j=w_j\sum_k\phi_{ik}\widetilde{q}_i/w_i \sum_k\phi_{jk},
$
\end{center}
we find that the convex optimization problem can actually be written in just one variable. Further, in the cases of social welfare and throughput, it is always the case that $\max_i\widetilde{q}_i=1$ for an optimal solution in the infinite-unit system. Hence, in these cases only one eigenvector computation is needed.

\paragraph{Discrete price set.}
We now show how the pricing policy from Algorithm \ref{alg:point_pricing} can be modified when there is a discrete set of available prices for each node (see Algorithm \ref{alg:point_pricing_rest}.  Let $\crl*{q_i^1,\dots, q_i^{k}}$ be the set of quantiles,  in decreasing order, corresponding to the available prices for node $i$. We assume that for all stations $i$ there exists an available quantile that is smaller-equal to the quantile that would be solved for in Algorithm \ref{alg:point_pricing}, i.e., there exist $q_i^{\ell}\leq \widetilde{q}_i'$. Then we obtain the quantiles $\widetilde{\q}$ by solving for the unconstrained case as in Algorithm~\ref{alg:point_pricing} to obtain preliminary quantiles $q'_i$ (not necessarily from the feasible set) and then setting each $\widetilde{q}_i$ to be the largest available quantile smaller or equal to $q'_i$ (see Algorithm \ref{alg:point_pricing_rest}). Our resulting guarantee has some extra loss that depends on how well the prices represent each part of the distribution. We now prove the performance guarantee for our policy. 

\begin{theorem}[approximation guarantee for discrete point price set]\vspace{0.1in}
\label{thm:point_discrete}
Consider any objective function $\obj_m$ for the $m$-unit system with concave reward curves $R_{ij}(\cdot)$.
Let $\widetilde{\p
\q}$ be the pricing policy returned by Algorithm \ref{alg:point_pricing_rest}, $\opt_m$ be the value of the objective function for the optimal state-dependent point pricing policy in the $m$-unit system. Suppose that there exists $\alpha$ such that for all $i$ and all $s$, $\alpha\cdot q_i^s\geq q_i^{s+1}$. Then,
$$
\alpha \obj_m(\widetilde{\q})
\geq \frac{m}{m+n-1}\opt_m,
$$
where $\opt_m$ is the objective of the optimal state-dependent policy for discrete prices in the $m$-unit system.
\end{theorem}
\begin{proof}
Let $\q'$ and $\widetilde{\q}$ be the quantiles as defined in Algorithm \ref{alg:point_pricing_rest}. Since $\widehat{\obj}(\q')$ is an upper bound on the unrestricted point pricing problem (see Theorem \ref{thm:point_pricing}), it is also an upper bound on $\opt_m$. Lemma \ref{lem:infinite} implies that $\widehat{\obj}(\q')=\obj_\infty(\q')$, since $\q'$ fulfills the balanced demand property. Further, by Lemma \ref{lem:finite}, $\obj_m(\widetilde{\q})\geq \frac{m}{m+n-1}\obj_\infty(\widetilde{\q})$. Thus, what remains is to bound $\obj_\infty(\widetilde{\q})$ with respect to $\obj_\infty(\widehat{q})$. Since $\widetilde{q}_i\leq\widehat{q}_i$ for all $i$ and the per-ride rewards $I_{ij}(\cdot)$ are assumed to be non-decreasing in the quantiles, we only need to bound the changes in the availabilities of the infinite-unit system for each $i$. Since the $w_i$ are constant under point-pricing, the availabilities are only affected by prices in the denominator, where the change is equal to $\widetilde{q}_i/\widehat{q}_i$. Thus, no traffic intensity changes by more than a factor of $\alpha$ and the result follows.
\end{proof}

We remark that the results here immediately extend to the rate-limiting constraint in Section \ref{ssec:delay}: if at least one $q_i'$ is rounded down, then the convergence must be linear; otherwise, no factor of $\alpha$ is incurred through rounding to feasible prices, but one can nonetheless round the quantiles down to the next-lowest quantiles. Either way, the bound from Lemma \ref{lem:Am_w_delay} applies.

\begin{algorithm}[!h]
\caption{The Elevated Flow Relaxation Program with a Discrete Set of Point Prices }
\label{alg:point_pricing_rest}
\begin{algorithmic}[1]
\REQUIRE arrival rates $\phi_{ij}$, value distributions $F_{i}$, reward curves~$R_{ij}$, feasible quantiles $\crl*{q_i^1,\ldots,q_i^k}$ in decreasing order.
\STATE Find $\crl*{q'_{i}}$ that solve the following point price relaxation:
$$
\begin{array}{rcll}
\max_\q &&  \sum_{(i,j)} \phi_{ij}R_{ij}({q}_{i}) &\\
\sum_{k} \phi_{ki}{q}_{k}  & = & \sum_{j} \phi_{ij}{q}_{i} & \forall\;i\\
{q}_{i} & \in & [0,1] & \forall\;i.
\end{array}
$$
\STATE Output \emph{state-independent} prices $\widetilde{p}_{i} = F_{i}^{-1}(1-
\max_j\crl*{q_i^j:q_i^j\leq q_i'})$ with corresponding quantiles $\widetilde{q}_{i}=\max_j\crl*{q_i^j:q_i^j\leq q_i'})$.
\end{algorithmic}
\end{algorithm}

\section{Settings Without Prices}\label{app:no_prices}

In Section \ref{ssec:rerouting} we discussed how two control levers, redirection of supply and of demand, can be combined with pricing to obtain the same guarantees we obtain for the pure pricing problem. We now show that our technique extends to settings in which only redirection of supply/demand is allowed, but pricing is not. We first present the results in the absence of travel delays and then consider delays as well. Intuitively, the key ingredient for both settings is to define variables, analogous to quantiles in the setting with prices, that capture the fraction of customers who receive service in the infinite-unit limit.

\subsection{Quantiles Without Prices}
Because we assume in this section that pricing is not allowed, $I_{ij}$ cannot be a function of prices (for exogenously set prices, $I_{ij}$ is constant). Thus, in this section, the elevated objective, defined analogously to Section \ref{ssec:elev_obj}, is always equal to the non-elevated objective. 

Similarly to Algorithm \ref{alg:point_pricing}, we introduce variables $q_i$ that capture the fraction of customers the elevated flow relaxation attempts to serve after some thinning. While we cannot implement an admission control through pricing, as is done in Section \ref{ssec:point_pricing}, we can implicitly ensure these to arise through availabilities. The variables $\q$ now correspond to the induced availabilities, i.e.,  one can think of the $q_i$ as $A_{i,\infty}(\r)$, in the infinite-unit limit; notice in particular that these used to always be equal to 1 when prices ensured that the balanced demand property be fulfilled. Intuitively, these thus correspond to effective quantiles in the infinite unit-limit (which would be the same as the theoretical quantile if we were pricing in a way the balanced demand property was fulfilled).

We adopt the same notation as in Section \ref{ssec:rerouting}, with the exception that we do not allow for pricing policies and thus everything is just a function of $\r$. Observe that the resulting flows are within the following polytope (as in Sections \ref{ssec:rerouting} and \ref{ssec:point_pricing}):
$$(1)\;\;\widehat{q}_{i}\in[0,1],\qquad
(2)\;\;\sum_k \prn*{\phi_{ki}\widehat{q}_{k} + \widehat{z}_{ki}}  = \sum_{j}\prn*{ \phi_{ij}\widehat{q}_{i}+\widehat{z}_{ij}},\qquad
(3)\;\;\sum_k \widehat{z}_{ik} \leq \sum_j \phi_{ji}\widehat{q}_{j}\;\forall\;i.$$
As in Section \ref{ssec:rerouting}, these constraints stem from demand bounding, supply circulation, and the limitation that only non-empty arriving vehicles may be rebalanced. The only difference in the polytope is that $\widehat{q}_{ij}$ were replaced by $q_{i}$ since an admission control that is based only on availability cannot differentiate between different destinations. Phrased differently, arriving customers have the same probability of a station having availability regardless of their destination.  Optimizing the elevated objective over the polytope given by these constraints is a linear program and yields an upper bound on the objective by the same arguments as in Lemma \ref{lem:state_dependent}. 
Consider the redirection policy $\widetilde{\r}$ obtained from the solution of the linear program (see Algorithm~\ref{alg:noprices}). In the next Lemma, we bound the infinite unit performance of this policy compared to the value of the elevated flow relaxation. 

\begin{algorithm}[!h]
\caption{The Elevated Flow Relaxation Program for Redirection without Prices}
\label{alg:noprices}
\begin{algorithmic}[1]
\REQUIRE arrival rates $\phi_{ij}$, per-ride rewards $I_{ij}$, rerouting costs $c_{ij}$. \\ 
\STATE Find $\{\widetilde{q}_{i},\widetilde{z}_{ij}\}$ that solve the following relaxation:
\[
\begin{array}{rcll}
\max_{\q,\z} &&  \sum_{i,j} \prn*{\phi_{ij}q_{i}I_{ij}-c_{ij}z_{ij}} &\\
\sum_k \prn*{\phi_{ki}q_{k} + z_{ki}}  & = & \sum_{j}\prn*{ \phi_{ij}q_{i} + z_{ij}} & \forall i\\
\sum_k z_{ik} & \leq & \sum_j \phi_{ji}q_{j} &\forall i\\
q_{i} & \in & [0,1] & \forall i
\end{array}
\]
\STATE Output redirection probabilities $\widetilde{r}_{ij} = \widetilde{z}_{ij}/\sum_{k}\phi_{ki}\widetilde{q}_{k}$
\end{algorithmic}
\end{algorithm}

\begin{lemma}\vspace{0.1in}
\label{lem:noprice_inf}
Denote by $\widetilde{\q}$ the quantiles solved for in the relaxation of Algorithm \ref{alg:noprices} and by $\widetilde{\r}$ the redirection probabilities returned. Then $\obj_\infty(\widetilde{\r})\geq \widehat{\obj}(\widetilde{\q},\widetilde{\r})$.
\end{lemma}
\begin{proof}
Consider first $\obj_\infty(\widetilde{\q},\widetilde{\r})$, the objective obtained when implementing both the redirection policy $\widetilde{\r}$ and the quantiles $\widetilde{\q}$ that Algorithm \ref{alg:noprices} solves for. By the same argument as in Lemma \ref{lem:reb_availability}, all availabilities are equal to 1 (and all traffic intensities are equal) in this system, and thus its objective matches $\widehat{\obj}(\widetilde{\q},\widetilde{\r})$. In order for us to compare $\obj_\infty(\widetilde{\q},\widetilde{\r})$ with $\obj_\infty(\widetilde{\r})$, consider a node $v\in\arg\max_j \widetilde{q}_j$. 
Increasing each quantile by a factor of $1/\widehat{q}_{v}$, we obtain quantiles $\bar{\q}$. Notice that in the system with quantiles $\bar{\q}$, the traffic intensity at each node is changed by the same factor, so the traffic intensities are still equal and the availabilities are still equal at every node.
In fact, for the relaxation in Algorithm \ref{alg:noprices}, there exists at least one $i$ such that $q_i=1$, so no quantile changes (when allowing for delays and scaling demand with the number of units as in Section \ref{ssec:delay}, this would not necessarily hold true).
Thus, $\obj_\infty(\bar{\q},\widetilde{\r})\geq \widehat{\obj}(\widetilde{\q},\widetilde{\r})$. Thereafter, for each node $j\neq v$, we increase its quantile to 1. Notice that each such change only decreases the traffic intensity at $j$, so the maximum traffic intensity remains unchanged. 
The lemma follows because the decrease in the traffic intensity (and thus availability) at each node $j\neq v$ is exactly balanced by the increased rate of arrivals at $j$. Formally, we have that $f_{jk,\infty}(\bar{\q},\widetilde{\r})$ remains unchanged when the $j$th coordinate of the quantiles is set to 1. Since the per-ride rewards are also assumed to remain unchanged (as in this section we assume $I_{ij}$ are constant, changing quantiles does not affect them), we then find $\obj_\infty(\widetilde{\r})=\obj_\infty(\bar{\q},\widetilde{\r})\geq \widehat{\obj}(\widehat{\q},\widetilde{\r})$.
\end{proof}

Now, using Lemma \ref{lem:noprice_inf} in place of Lemma \ref{lem:reb_availability} in the proof of Theorem \ref{thm:rerouting}, we get the following.

\begin{theorem}[demand redirection without pricing]\vspace{0.1in}
onsider an for the $m$-unit system. Let $\widetilde{\q}$ and $\widetilde{\r}$ be as in Algorithm~\ref{alg:noprices}, and $\opt_m$ be the objective of the optimal state-dependent policy in the $m$-unit system. Then
$$
\obj_m(\widetilde{\r}) \geq \frac{m}{m+n-1}\opt_m.
$$
\end{theorem}
To illustrate the discussion at the end of Section \ref{sec:pricing}, we want to point out that $\obj_m(\widetilde{\r})$ above is not based on a policy that fulfills the balanced demand property; thus the bound between $\obj_m(\widetilde{\r})$ and $\obj_\infty(\widetilde{\r})$ cannot be obtained by factoring out $A_{i,m}(\widetilde{\r})$, and is instead based on the infinite-unit limit.

\subsection{Delays Without Prices} 

Accommodating settings in which we are not allowed pricing, but do have delays, requires an additional idea. This is because the argument in Section \ref{ssec:delay} explicitly relied on pricing to ensure that (on average) not too many units are in transit simultaneously, thereby enabling a lower bound on the maximum availability. 
Without prices to regulate demand, we can no longer control the maximum availability. Instead, we use the following stochastic dominance characterization for closed-queueing networks.

\begin{lemma}[Theorem 3.8 in Chen and Yao~\cite{chen2013fundamentals}]\vspace{0.1in}
\label{thm:chen}
In a closed Jackson network, with state-independent service rates, increasing the service rate functions,
in a pointwise sense, at a non-empty subset of nodes increases throughput. 
\end{lemma}

In our context, this is equivalent to saying that increasing quantiles at a subset of nodes only increases throughput. In fact, one can show that throughput also increases locally, i.e., increasing quantiles at one node (which we henceforth refer to as \emph{point quantiles}) does not decrease the rate of units on any edge.

\begin{lemma}\vspace{0.1in}
\label{corollary:chao}
Let $\q=\{q_i\}$ be a vector of point quantiles, and $\widetilde{\q}$ be a vector of point quantiles with $\widetilde{q}_k\geq q_k \,\forall\, k$. Then for any pair $(i,j)$, we have $f_{ij,m}(\q)\leq f_{ij,m}(\widetilde{\q})$, i.e. the rate of realized trips from $i$ to $j$ does not decrease when point quantiles are increased.
\end{lemma}

\begin{proof}
The proof relies on two observations. Note first for $\q$ and $\widetilde{\q}$, we have $$\frac{\phi_{ij}q_i}{\sum_k \phi_{ik}q_i}=\frac{\phi_{ij}\widetilde{q}_i}{\sum_k \phi_{ik}\widetilde{q}_i}\,\forall\, i,j,
$$
and therefore, letting $\w(\q')$ denote the eigenvector of the routing matrix $\{\phi_{ij}(q_i')/\sum_k\phi_{ik}(q_i')\}_{i,j\in[n]^2}$ (see Section \ref{ssec:point_pricing}), we obtain $w_i(\widetilde{\q})=w_i(\q)$. Define $\Gamma_m(\q)\triangleq G_m(\q)/G_{m-1}(\q)$. We now have that the ratio of the rates $f_{ij,m}(\q)/f_{ij,m}(\widetilde{\q})$ is equal to
$$
\frac{f_{ij,m}(\q)}{f_{ij,m}(\widetilde{\q})} = \frac{A_{i,m}(\q)q_i\phi_{ij}}{A_{i,m}(\widetilde{\q})\widetilde{q}_i\phi_{ij}}
=\frac{\Gamma_m(\q)r_i(\q)q_i}{\Gamma_m(\widetilde{\q})r_i(\widetilde{\q})\widetilde{q}_i}
=\frac{\Gamma_m(\q)\frac{w_i(\q)}{\sum_k\phi_{ik}q_i}q_i}{\Gamma_m(\widetilde{\q})\frac{w_i(\widetilde{\q})}{\sum_k\phi_{ik}\widetilde{q}_i}\widetilde{q}_i}=\frac{\Gamma_m(\q)}{\Gamma_m(\widetilde{\q})}.
$$
Note that the ratio of $f_{ij,m}(\q)$ and $f_{ij,m}(\widetilde{\q})$ does not depend on $i$ and $j$.
Moreover, from Lemma \ref{thm:chen} we have
$\sum_{i,j} f_{ij,m}(\q)\leq  \sum_{i,j}f_{ij,m}(\widetilde{\q})$. 
Combining the two, we get $f_{ij,m}(\q)\leq f_{ij,m}(\widetilde{\q})$.
\end{proof}

This allows us to prove the guarantee of Theorem \ref{thm:pricing_w_delays} for settings in which prices cannot be used to provide a lower on the maximum availability within the system.

\begin{algorithm}[!h]
\caption{The Rate-Limited Elevated Flow Relaxation Program for Redirection w/o Prices}
\label{alg:braverman}
\begin{algorithmic}[1]
\REQUIRE scaling paramter $\epsilon_m$, arrival rates $\phi_{ij}$, rewards $I_{ij}$, rerouting costs $c_{ij}$, travel-times $\tau_{ij}$. \\ 
\STATE Find $\{\widetilde{q}_{i},\widetilde{z}_{ij}\}$ that solve the following relaxation:
\[
\begin{array}{rcll}
\max_{\q,\z} &&  \sum_{i,j} \prn*{\phi_{ij}q_{i}I_{ij}-c_{ij}z_{ij}}\\
\sum_{i,j}\phi_{ij}\tau_{ij}q_i+z_{ij}  & \leq & m\\
\sum_k \prn*{\phi_{ki}q_{k} + z_{ki}}  & = & \sum_{j}\prn*{ \phi_{ij}q_{i} + z_{ij}} & \forall i\\
\sum_k z_{ik} & \leq & \sum_j \phi_{ji}q_{j} &\forall i\\
q_{i} & \in & [0,1] & \forall i
\end{array}
\]
\STATE Output redirection probabilities $\widetilde{r}_{ij} = \widetilde{z}_{ij}/\sum_{k}\phi_{ki}\widetilde{q}_{k}$
\end{algorithmic}
\end{algorithm}

\begin{theorem}\vspace{0.1in}
Consider any objective for the $m$-unit system. Let $\widetilde{\q}$ and $\widetilde{\r}$ be as in Algorithm~\ref{alg:braverman}  with $\varepsilon_m\defeq2\sqrt{\ln m/m}$, and $\opt_m$ be the objective of the optimal state-dependent policy in the $m$-unit system and $m\geq100$. \begin{equation*}
\frac{\obj_m(\widetilde{\r})}{\opt_m}\geq \prn*{1-\varepsilon_m}\prn*{\frac{\sqrt{m\ln m}}{\sqrt{m\ln m}+n-1} - \frac{3}{\sqrt{m\ln m}}}. 
\end{equation*}
\end{theorem}

\begin{proof}
The same proof as in Theorem \ref{thm:pricing_w_delays} guarantees that using point prices as given by $\q(1-\epsilon_m)$, where $\q$ comes from the solution of the relaxation in Algorithm \ref{alg:braverman} yields the required guarantee. Lemma \ref{corollary:chao} then guarantees that increasing all quantiles to one yields a solution no worse.
\end{proof}

\noindent We remark that with $m\to\infty$, the above theorem recovers and strengthens through a finite guarantee, and a rate of convergence, the result of Braverman et al \cite{braverman2016empty}. In fact, the same reasoning gives an analogue of Corollary \ref{cor_lin_rate} which (for some regimes) even gives a linear rate of convergence to the objective of the elevated flow relaxation.

Finally, we note that Lemma \ref{corollary:chao} also yields an alternate proof of Lemma \ref{lem:finite}: given quantiles $\q$ that do not induce balanced demand, we consider a system with rates $\widetilde{\phi}_{ij}=\phi_{ij}\frac{\max_k r_k(\q)}{r_i(\q)}$. We observe that (i) the objectives with rates $\widetilde{\phi}_{ij}$ and rates $\phi_{ij}$ are the same in an infinite unit system and (ii) that the system with rates $\widetilde{\phi}_{ij}$ obeys the balanced demand property. Thus, the counting argument of \cite{whitt1984} guarantees an objective within $m/(m+n-1)$  of the infinite unit system in a system with rates $\widetilde{\phi}_{ij}$. However, by (i) the latter was equal to the upper bound on $\opt_m$. Since Lemma \ref{corollary:chao} implies that the $m$-unit system with rates $\phi_{ij}$ has objective no worse than the $m$-unit system with rates $\widetilde{\phi}_{ij}$, the statement of the lemma follows.

\section{Auxiliary Lemma}\label{app:proofs_point}
We present a basic Chernoff tail bound for Poisson random variables, which we use in Section~\ref{ssec:delay}
\begin{lemma}
\label{lem:poissontail} \vspace{0.1in}
For $X\sim\,$Poisson$(\lambda)$, we have for any $0\leq x\leq \lambda$:
$$\PP[X>\lambda + x]\leq \exp\prn*{-\frac{x^2}{2\lambda}\prn*{1 - \frac{x}{\lambda}}}$$
\end{lemma}
\begin{proof}
Using a standard Chernoff bound argument, we have for any $\theta\geq 0$:
\begin{align*}
\PP[X>\lambda+x]
&= \PP[e^{\theta X}>e^{\theta(\lambda+x)}]
\leq
e^{-\theta(\lambda+x)}\EE\brk*{e^{\theta X}} = e^{-\theta(\lambda+x)}\cdot e^{\lambda(e^{\theta}-1)}
\end{align*}
Now, optimizing over the choice of $\theta$, we get
\begin{align*}
\PP[X>\lambda+x]
&\leq\exp\prn*{\inf_{\theta}\prn*{\lambda\prn*{e^{\theta}-1-\theta}-x\theta}}\\
&=\exp\prn*{x-(x+\lambda)\log\prn*{1+x/\lambda}}\quad\mbox{(Setting $\theta=\log\prn*{1+x/\lambda}$)}\\
&\leq \exp\prn*{x-(x+\lambda)\prn*{\frac{x}{\lambda} - \frac{x^2}{2\lambda^2}}}
=\exp\prn*{-\frac{x^2}{2\lambda}\prn*{1 - \frac{x}{\lambda}}} 
\end{align*}
\end{proof}

\section{Alternate Proof of  Lemma \ref{lem:finite}}\label{app:zahorjan_proof}
As we mentioned in the introduction, Lemma \ref{lem:finite} can also be proven using stochastic coupling arguments. One such argument was given by Zahorjan et al. \cite{zahorjan1982balanced} to obtain an inequality (Inequality 6 in \cite{zahorjan1982balanced}) that, translated into our context, guarantees
$$
\obj_{m}^{T}(\q) \geq \frac{m}{(m+n-1)r_{\max}(\q)}.
$$
Since the left-hand side is equal to $$\sum_{i=1}^n A_{i,m}(\q)\phi_{ij}q_{ij}=\sum_{i=1}^n \frac{G_{m-1}(\q)}{G_m(\q)}w_i(\q),$$
we can derive (normalizing with $\sum_i w_i(\q)=1$) that
$$
\frac{G_{m-1}(\q)}{G_m(\q)}\geq \frac{m}{(m+n-1)r_{\max}(\q)} ,
$$
which, by Lemma \ref{lem:max_availability}, implies the result.

\end{document}